\documentclass[preprint,12pt]{elsarticle}

\usepackage{amssymb,amsthm,latexsym}

\theoremstyle{plain}
\newtheorem{theorem}{Theorem}[section]
\newtheorem{lemma}[theorem]{Lemma}

\newtheorem{prop}[theorem]{Proposition}

{\theoremstyle{definition}

}

\usepackage{paralist}
\journal{Finite Fields and Their Applications}

\begin{document}

\begin{frontmatter}



\title{Some binary BCH codes with length $n=2^m+1$}

\author[a,b]{Yang Liu}
\author[a]{Ruihu Li$^*$ }
\ead{liu\_yang10@163.com, llzsy110@163.com}
\author[a]{ Qiang Fu}
\author[a]{Liangdong Lu}
\author[a]{Yi Rao}
\address[a]{\small Department of Basic, Air Force Engineering University\\
\small Xi'an, Shaanxi 710051, P. R. China}
\address[b]{\small The first Aviation College, Air Force Engineering University\\
\small Xinyang, Henan 464000, P. R. China}
\begin{abstract}
Under research for near sixty years,
Bose-$\!$Ray-$\!$Chaudhuri-$\!$Hocquenghem(BCH) codes have played
increasingly important roles in many applications such as
communication systems, data storage and information security.
However, the dimension and minimum distance of BCH codes are seldom
solved  until now because of their intractable characteristics. The
objective of this paper is to study the dimensions of some  BCH
codes of length $n=2^m+1$ with $m=2t+1, 4t+2, 8t+4$ and $m\geq 10$.
Some new techniques   are employed to investigate coset leaders
modulo $n$. For each type of $m$ above, the first five largest coset
leaders modulo $n$ are determined, the dimension  of some BCH codes
of length $n$ with designed distance $\delta>2^{\lceil \frac{m}{2}
\rceil}$ is presented. These new techniques and results may be
helpful to study other families of cyclic codes over finite fields.
\end{abstract}


\end{frontmatter}


\section{Introduction}
\label{sec1} Let $q$ be a prime power,   $\mathbb{F}_{q}$ be the
finite field with $q$ elements.  A  linear $[n,k,d]$ code
$\mathcal{C}$ over $\mathbb{F}_{q}$  is a $k$-dimensional subspace
of $\mathbb{F}^{n}_{q}$ with minimum  (Hamming) distance $d$. A
linear code $\mathcal{C}$ is cyclic if
  $(c_{0},c_{1},\cdots,c_{n-1}) \in \mathcal{C}$  implies
  $(c_{n-1},c_{0},c_{1},\cdots,c_{n-2})$ $\in$ $\mathcal{C}$.
There is a bijective correspondence between  vectors $
(c_{0},c_{1},\cdots,c_{n-1})$ $\in \mathbb{F}^{n}_{q}$ and
  polynomials $ c_{0}+c_{1}x+\cdots+c_{n-1}x^{n-1}$ $\in \mathbb{F}_{q}[x]/(x^{n}-1),$
  then a cyclic code $\mathcal{C}$ is identified with an ideal of
$\mathbb{F}_{q}[x]/(x^{n}-1)$. Since each ideal of
$\mathbb{F}_{q}[x]/(x^{n}-1)$ is  principal, a cyclic  code
$\mathcal{C}$ $=(g(x))$, where $g(x)$ is monic and has the smallest
degree among all the generators of $\mathcal{C}$, and $g(x)$ is
called the {\it generator polynomial} of  $\mathcal{C}$.

If $\gcd(q,n)=1$, a {\it $q$-cyclotomic coset} modulo $n$ containing
$x$ is defined by
$$C_{x}=\{{x,xq,xq^{2},...,xq^{e-1}}\} \bmod  n,$$
where $e$ is the smallest positive integer such that $q^{e}x\equiv x
 \bmod   n$,  the cardinality of $C_{x}$ is denoted by $|C_{x}|=e$.
The smallest integer in $C_{x}$ is called  the {\it coset leader} of
$C_{x}$ modulo $n$.

Let $\gcd(q,n)=1$. If $\xi$ is a primitive $n$-th root of unity in
some finite field containing $\mathbb{F}_{q}$, $T$$=\{i \mid
g(\xi^{i})=0\}$ is called the {\it defining set} of  $\mathcal{C}$
$=(g(x))$. It is well known that $T$  is  the union of some
$q$-cyclotomic   cosets modulo $n$, the dimension $k$ of
$\mathcal{C}$ is determined by $k=n-|T|$, and the minimum distance
$d$ is also determined by $T$. If $T=C_{b}\cup C_{b+1}\cup \cdots
\cup C_{b+\delta-2}$, then $\mathcal{C}$ is called a BCH code of
designed distance $\delta$, and it can be denoted by
$\mathcal{C}(n,q,\delta,b)$ as \cite{Ding5}-\cite{Ding8}. If $b=1$,
$\mathcal{C}$ is called   {\it narrow-sense};  otherwise,   {\it
non-narrow-sense}. If $n=q^{m}-1$, $\mathcal{C}$ is called  {\it
primitive};
  otherwise {\it nonprimitive}.
Particularly, if $n=q^{m}+1$,  it is called an  {\it antiprimitive}
BCH code by  Ding in \cite{Ding5}.

 Since   binary BCH codes were
discovered  independently around 1960 by Hocquenghem
\cite{Hocquenghem},  Bose and Ray-Chaudhuri
\cite{Bose1},\cite{Bose2}, they have been studied and employed
widely in practice \cite{Ding5}.  However,  limited knowledge about
BCH codes is still  limited. For example, for the  given finite
field, the dimension and minimum distance of a BCH code are known
only for special code lengths and designed distances, while open in
general. As pointed out by Charpin in \cite{Char} and Ding in
\cite{Ding5}, it is notoriously hard  to generally determine
 the dimension  and minimum distance  of BCH codes. For details on
  advancement about BCH codes, please see \cite{Mann}-\cite{Ding8}
   and the references therein.

 For general $n$ and $\delta$,  little is known about the parameters
of $\mathcal{C}(n,q,\delta,b)$, only lower bounds on the dimension
and minimum distance have been developed, see
\cite{Mann}-\cite{Ding8} and the next section. Primitive BCH codes
are the most studied among all types of BCH codes. Ding {\it et al.}
deeply investigated parameters of primitive BCH codes, determined
their dimension   and gave very good bounds on the minimum distance
of   codes for
  relatively small $\delta$ and  special types of $\delta$
\cite{Ding1}-\cite{Ding8}. They also determined the first, the
second  and third largest $q$-cyclotomic coset leaders $\delta_{1}$,
$\delta_{2}$, $\delta_{3}$ modulo $n=q^{m}-1$, and completely
presented the dimensions, minimum distances and weight distributions
of $\mathcal{C}(n,q,\delta,1)$ and $\mathcal{C}(n,q,\delta+1,0)$ for
$\delta=\delta_{i}$ with $i=1,2,3$. For more details, please see
\cite{Ding4} and the references therein.

 Recently, scholars have paid  much attention to antiprimitive  BCH
codes $\mathcal{C}(n,q,\delta,b)$ for $n=q^m+1$. As pointed by
\cite{Ding7} \cite{Ding8}, antiprimitive BCH codes  are all linear
codes with complementary dual (LCD) and contain many good codes. LCD
codes  can be used against side-channel attacks and fault
noninvasive attacks \cite{Carlet}.
 Ding, Li and  Liu {\it et al.} acquired some achievements about the
parameters of    BCH codes of length $n=q^{m}+1$ for designed
distance $\delta \leq q^{\lfloor \frac{m-1}{2} \rfloor}+3$ in
\cite{Ding7} and  $\delta \leq q^{\lceil \frac{m}{2} \rceil}$ in
\cite{Ding8}, respectively. Ding pointed out that it is very
significative  to find the second and third largest coset leaders
modulo $n=q^m+1$, which is helpful to deduce parameters of BCH codes
\cite{Ding5}.

The main objective of this paper is to determine the first five
largest 2-cyclotomic coset leaders  modulo $n=2^m+1$   and calculate
the dimensions of $\mathcal{C}(2^m+1,2,\delta,1)$ and
$\mathcal{C}(2^m+1,2,\delta+1,0)$ for a larger range of $\delta$
when $m=2t+1, 4t+2, 8t+4$ and $m\geq 10$. In Section 2, basic
concepts on LCD codes, some known results on BCH codes  are
reviewed. In Sections 3, 4 and 5, the parameters of BCH codes of
length $n=2^m+1$ with designed distance $\delta>2^{\lceil
\frac{m}{2} \rceil}$ are determined for $m=2t+1$, $m=4t+2$ and
$m=8t+4$, respectively. The final conclusion is drawn in Section 6.

\section{Preparation }
\label{sec2}
 In this section, we recall the  basic knowledge on  BCH codes  and LCD codes.
  For more details, one can refer to Refs. \cite{mac} and \cite{huf}.

   If $\mathcal{C}$ is a linear code of length $n$ over
$\mathbb{F}_{q}$, its Euclidean dual code is defined by
$$\mathcal{C}^{\perp }=\{ X \in
\mathbb{F}_{q}^n \mid(X, Y)=X Y^{T}=0 \ \hbox{for all}\ \ Y \in
\mathcal{C} \},$$  where $Y^{T}$ denotes  the transposition of the
vector $Y=(y_{1},y_{2},..., y_{n})$.
  A linear code $\mathcal{C}$ is called an {\it LCD} code  if $\mathcal{C}^{\perp
}\bigcap\mathcal{C}=\{\bf 0\}$, which is equivalent to
$\mathcal{C}^{\perp }\bigoplus\mathcal{C}=\mathbb{F}_q^n$.

A linear code $\mathcal{C}$ is called  {\it reversible} if
$(c_{0},c_{1},\cdots,c_{n-1}) \in \mathcal{C}$ implies that
$(c_{n-1},c_{n-2},\cdots,c_{0}) \in \mathcal{C}$. Hence, a cyclic
code  $\mathcal{C}$ with generator polynomial $g(x)$ is LCD if and
only if $g(x)$ is self-reciprocal or $\mathcal{C}$ is reversible.

BCH codes have a long history  and are   the best linear codes known
when their lengths are moderate. However, it is challenging to
determine their parameters in general. The following lower bounds on
the dimension of  BCH codes are well known, see \cite{mac},
\cite{huf}.
\begin{prop}\label{prop2.1}
  Let $k$ be the dimension of
$\mathcal{C}(n,q,\delta,b)$. Then

(1) $k \geq n-ord_{n}(q)(\delta-1)$.

(2) If $q=2$, $b=1$ and $\delta$ is odd, then $k \geq
n-ord_{n}(q)(\delta-1)/2$.
\end{prop}
The above bounds are tight only when $\delta$ is very small, while
they will become  negative when $\delta$ gets larger. In \cite{Aly},
a formula on $k$ $=\dim(\mathcal{C}(n,q,\delta,1))$ was obtained,
which generalized that of \cite{Yue}.
\begin{prop}\label{prop2.2}
  Let  $\gcd(q,n)=1$, $q^{\lfloor\frac{m}{2}\rfloor}\leq n\leq
q^{m}-1$ and $m=ord_{n}(q)$. If $2\leq \delta\leq
\min\{nq^{\lceil\frac{m}{2}\rceil}/(q^{m}-1),n \}$, then the
dimension $k$ of $\mathcal{C}(n,q,\delta,1)$ is
$k=n-m\lceil(\delta-1)(1-1/q)\rceil.$
\end{prop}
 This result is useful in some
cases, such as primitive BCH codes, while useless for BCH codes with
length $n=q^{m}+1$. Cyclic codes with length $n=q^{m}+1$ are LCD
codes and  contain many good linear codes, see \cite{Ding7},
\cite{Ding8} and \cite{Rao}.  Ref. \cite{Rao} discussed binary LCD
cyclic codes of length $n \leq 1025$.  Known results on the
dimensions  of BCH codes  $\mathcal{C}(q^{m}+1,q,\delta,1)$ and $
\mathcal{C}(q^{m}+1,q,\delta+1,0)$   in \cite{Ding7} and
\cite{Ding8} are as follows.

\begin{lemma}\label{lemm2.3}( Theorem 18 of \cite{Ding7} )
For any integer $\delta$ with $3 \leq \delta \leq  q^{\lfloor
\frac{l-1}{2}\rfloor }+3$, the reversible code
$\mathcal{C}(q,n,\delta,0)$ has parameters
$[q^l+1,q^l-2l(\delta-2-\lfloor \frac{\delta-2}{q}\rfloor), d\geq
2(\delta-1)]$.
 \end{lemma}

\begin{lemma}\label{lemm2.4} (Theorems 38 and
39 of \cite{Ding8}) Let $m\geq 3$ be an integer  and $h=\lfloor
\frac{m-1}{2}\rfloor$.  Then  the following hold for $2 \leq \delta
\leq q^{h+1}$:

(1) If $m\geq 4$ is an even integer, then
$\mathcal{C}(n,q,\delta,1)$ has parameters
$[q^m+1,q^m+1-2m(\delta-1-\lfloor \frac{\delta-1}{q}\rfloor),d\geq
\delta]$ and $ \mathcal{C}(n,q,\delta+1,0)$ has parameters
$[q^m+1,q^m-2m(\delta-\lfloor \frac{\delta}{q}\rfloor),d\geq
2\delta]$.

(2) If $m\geq 3$ is an odd integer, $2 \leq \delta \leq q^{h+1}$,
$k=\dim(\mathcal{C}(n,q,\delta,1))$, then
$$k=\left\{
\begin{array}{lll}
 q^m+1-2m(\delta-1-\lfloor \frac{\delta-1}{q}\rfloor)  &\mbox {if $\delta\leq q^{h+1}-q$;}\\
q^m+1-2m(q^{h+1}-q-\lfloor \frac{\delta-1}{q}\rfloor) &\mbox {if $
q^{h+1}-q+1 \leq \delta\leq q^{h+1}$.}
 \end{array}
 \right.$$
 \end{lemma}

 In the sequel,  we  only consider  BCH codes of length $n=2^m+1(m\geq10)$
 with   defining sets $T_{\delta}=C_{1}\cup C_{2}\cup \cdots \cup C_{\delta-1}$ and
$T_{\delta}'=C_{0}\cup C_{1}\cup  \cdots \cup C_{\delta-1}$ for odd
$\delta$, which can be denoted by $\mathcal{C}(n,2,\delta,1)$ and
$\mathcal{C}(n,2,\delta+1,0)$ as Refs. \cite{Ding7} and
\cite{Ding8}. Hence, $\mathcal{C}(n,2,\delta,1)$ $=[n,k,d]=[n,
n-\mid T_{\delta}\mid,d\geq \delta]$ and
$\mathcal{C}(n,2,\delta+1,0)$ $=[n,k_{0},d_{0}]=[n, 2^m-\mid
T_{\delta}'\mid,d_{0}\geq 2\delta]$.

Our approach to   the dimension  of   BCH codes is to find out all
coset leaders in $T_{\delta}$ and calculate their cardinalities.
While we find out  coset leaders in $T_{\delta}$ or $T_{\delta}'$,
some new skills are developed,   which are not employed in the
literature. For convenience, we first present three notations below.

{\bf Notation 1.} Let $x\equiv y$ denote $x\equiv y \bmod n$ unless
otherwise stated. We omit the words ``modulo $n$"  when we mention
``cyclotomic coset", ``coset leader" and so on. ``$x$ is a coset
leader" means that $x$ is a coset leader of $C_x$ modulo $n$.

 {\bf Notation 2.} Assume that $a,b,c$ are all nonnegative integers and $a\leq b$,
 denote $[a,b]=\{a,a+1,\cdots,b\}$ and   $[a,b]+c=[a+c,b+c]$.

 {\bf Notation 3.} Suppose that $C_{x}$  is a  $2$-cyclotomic coset containing
 $x$.
  $C_{x}$ can be given by $C_{x}=\{y_{_{x,0}},y_{_{x,1}},y_{_{x,2}},...,y_{_{x,2m-1}}\},$
  where $$y_{_{x,k}}\equiv 2^k x~\hbox{and}~y_{_{x,k}}\in \mathbb{Z}_n=\{0,1,2,\cdots,n-1.\}$$
 When $m\leq k\leq 2m-1$,  $ y_{_{x,k}}\equiv2^k x=2^{k-m}(2^{m}+1-1)x \equiv -2^{k-m}x \equiv
 n-y_{_{x,k-m}}$. Hence $C_{x}$ can be also denoted by $$\{y_{_{x,k}},n-y_{_{x,k}}| 0\leq k\leq m-1\}.$$

 It is easy to
see that $x$ is a coset leader of $C_{x}$ if and only if
$y_{_{x,k}}-x \geq 0$ and $n-y_{_{x,k}}-x \geq 0$ for $0\leq k \leq
m-1$. Obviously, if  $x$ is even, then $\frac{x}{2}\in C_x$ and $x$
is not a coset leader. Hence, to find out coset leaders in
$T_{\delta}$ or $T_{\delta}'$, it suffices to consider an odd $x$.

  For example, if $n=2^3+1=9$, we have then

  $C_{1}=\{y_{_{1,0}},y_{_{1,1}},y_{_{1,2}},y_{_{1,3}},y_{_{1,4}},y_{_{1,5}}\}=\{1,2,4,8,7,5\}$,

  $C_{3}=\{y_{_{3,0}},y_{_{3,1}},y_{_{3,2}},y_{_{3,3}},y_{_{3,4}},y_{_{3,5}}\}=\{3,6,3,6,3,6\}=\{3,6\}$,

  $C_{5}=\{y_{_{5,0}},y_{_{5,1}},y_{_{5,2}},y_{_{5,3}},y_{_{5,4}},y_{_{5,5}}\}=\{5,1,2,4,8,7\}$,\\
which imply that 1 and 3 are coset leaders, while 5 and 7  not.

Inspired by these latter work, by a detailed analysis of binary
cyclotomic cosets modulo $n=2^{m}+1$ with $m\not \equiv 0 \bmod 8$,
we deduce the first five largest coset leaders modulo $n$ and
determine the coset leader of $C_x$ for each
$$x\leq \left\{
\begin{array}{lll}
2^{t+2}+7
   &\mbox {if $m=2t+1$};\\
 2^{2t+2}+2^{2t+1}+3 &\mbox {if $m=4t+2$};\\
 2^{4t+3}+2^{4t+2}+2^{4t+1}+1 &\mbox {if $m=8t+4$,}
 \end{array}
 \right.$$ along with  their cardinalities.
  Thus, parameters of some binary antiprimitive  BCH codes with designed distance
$\delta>2^{\lceil \frac{m}{2} \rceil}$ are given, which widely
extend the range of $\delta$  in  \cite{Ding7} and \cite{Ding8} for
$q=2$.

\section{ BCH Codes of length $n=2^{m}+1$ for $m=2t+1$}
\label{sec3} Throughout this section,  we fix $n=2^{m}+1$ with
$m=2t+1\geq 11$.

It will be first determined when $x$ is a coset leader  for  $1\leq
x \leq 2^{t+2}+7$ in Theorem 3.1. Next, we show each $\delta_{i}$ is
a coset leader  for $1\leq i \leq 5$ in Lemma  3.2.  For verifying
the first five largest coset leaders in Theorem 3.5, we will
continue to give  Lemmas 3.3-3.4 on basis of Lemma 3.2. The
cardinalities of relevant cyclotomic cosets will be calculated in
Lemma 3.6. At last, Theorem 3.7 can be further given, which lists
dimensions of some BCH codes of length $n$ with designed distance
$\delta$ for $2^{t+1}+3\leq \delta \leq 2^{t+2}+9$ and
$\delta_5+2=\frac{n-3}{6}-8 \leq \delta \leq n$.

For clarity,  the proofs of Theorem 3.1,
 Lemma  3.2 and Theorem 3.7 will be  presented in {\bf Appendixes  A,   B  \hbox{and} C},
 respectively.

 \begin{theorem}\label{ther3.1}Let  $x$ be
odd. Then  we have the following:

(1) If $1\leq x\leq 2^{t+1}-3$, then $x$ is a coset leader, see
Refs. \cite{Ding7} and \cite{Ding8}.

(2) If  $2^{t+1}+3\leq x\leq 2^{t+1}+2^{t}-3$, then $x$ is a coset
leader.

(3) If $2^{t+1}+2^{t}+3\leq x\leq 2^{t+2}-9$, then $x$ is a coset
leader.

(4) If $x=2^{t+1}-1, 2^{t+1}+1$, $2^{t+1}+2^{t}-1, 2^{t+1}+2^{t}+1$
or  $ 2^{t+2}-7\leq x\leq 2^{t+2}+7$, then  $x$ is not a coset
leader.

\end{theorem}

 \begin{lemma}\label{lemm3.2}
Suppose $\delta_{1}=\frac{n}{3}$, $\delta_{2}=\frac{n-3}{6}$,
$\delta_{3}=\delta_{2}-2$, $\delta_{4}=\delta_{2}-8$,
$\delta_{5}=\delta_{2}-10$. Then $\delta_{1}, \delta_{2},
\delta_{3}$, $\delta_{4}$ and $\delta_{5}$  are all coset leaders.
\end{lemma}

Next, we will show $\delta_i( i=1,2,3,4,5)$  are the first five
largest coset leaders modulo $n$, respectively. It is necessary to
verify
 $x$ is not a coset leader if $\delta_{5}<x<n$ and $x\neq
 \delta_{i}$ for $1\leq i \leq 4$.
To achieve this, we first introduce an iterative algorithm to
partition  $ I^{(t)}=[1, 2^{2t-5}]$ into $2^{t-3}$ subintervals $$
I^{(t)}=[1, 2^{2t-5}] =I_1\bigcup I_2 \cdots \bigcup I_{2^{t-3}},
\hbox{for}~t\geq 5.$$
 \noindent{\bf Iterative Algorithm 1 (IA 1):}

Suppose  $I_i=[a_i,b_i]$, a  partition of $I^{(t)}=[1, 2^{2t-5}]$ is
obtained as
 $$I^{(t)}=I^{(t-1)}\bigcup (\bigcup \limits_{u=2^{t-4}+1}^{2^{t-3}} I_{u}), $$
$$\hbox{where} ~I_{u}=\left\{
\begin{array}{lll}
I_{j}+2^{2\times (t-4)+1}  &\mbox {if $u=j+2^{t-4}, j \in[1, 2^{t-4}-1]$;}\\
\mbox{[}a_{2^{t-4}}+2^{2t-7},2^{2t-5}\mbox{]} &\mbox {if $
u=2^{t-3}$,}
 \end{array}
 \right.$$
with the initial conditions: $I_1=[1,2],
I_2=[a_1+2^1,2^{2\times1+1}]=[3,8]. $

 For clarity, we give the partition in detail as follow.

1) If $t=5$, then $ I^{(5)}=[1, 2^{5}]$ can be partitioned into
$2^{2}$ subintervals as follows:

 $I_1=I_{2^0}=[1,2^1]=[1,2]=[a_1,b_1]$,

$I_2=I_{2^1}=[a_1+2^1,2^{2\times1+1}]=[3,8]=[a_2,b_2]$,

$I_3=I_{1+2^1}=I_1+2^{2\times1+1}=I_1+2^3=[9,10]=[a_3,b_3]$,

 $I_4=I_{2^2}=[a_{2^{1}}+2^{2\times 1+1},
2^{2\times 2+1}]=[a_2+2^3,2^5]=[11,32]=[a_4,b_4]$.

2) If $t=6$, then $ I^{(6)}=[1, 2^{7}]$ can be partitioned into
$2^{3}$ subintervals as follows:

 Let $I_{s}=[a_s,b_s]$ for $1\leq
s\leq 4$ as be given in 1),
 as for $2^2<s=j+2^2\leq 2^{2+1}$ with $1\leq j\leq 2^2-1$,
 define

 $I_s=I_j+2^{2\times 2+1}=I_j+2^{5}=[a_s,b_s]$ for $5\leq s\leq7$, and define

  $I_8=I_{2^3}=[a_{2^{2}}+2^{2\times 2+1}, 2^{2\times 3+1}]=[a_4+2^5, 2^7]=[a_8,b_8]$.

3) Let $t\geq 7$.  When a partition of $ I^{(t-1)}$ is given by

 $I^{(t-1)}=[1, 2^{2t-7}]$ $=I_1\bigcup I_2 \cdots \bigcup
I_{2^{t-4}}=\bigcup \limits_{j=1}^{2^{t-4}} I_j$, where
$I_{j}=[a_j,b_j]$.

  For $u=j+2^{t-4}$ with $1\leq j\leq 2^{t-4}-1$, define

$I_{u}=I_{j}+2^{2\times (t-4)+1}$
$=I_{j}+2^{2t-7}=[a_j+2^{2t-7},b_j+2^{2t-7}].$

For $u=2^{t-3}$,   define

 $I_{2^{t-3}}=[a_{2^{t-4}}+2^{2\times
(t-4)+1}, 2^{2\times (t-3)+1}]=[a_{2^{t-4}}+2^{2t-7},2^{2t-5}].$

Then, a partition of $I^{(t)}=[1, 2^{2t-5}]$ is  iteratively obtained as\\
 $I^{(t)}$ $=I^{(t-1)} \bigcup (\bigcup
\limits_{u=2^{t-4}+1}^{2^{t-3}} I_{u})$
  $=(I_{1}\cup I_{2} \cdots
\cup I_{2^{t-4}})\cup I_{2^{t-4}+1} \cdots \cup I_{2^{t-3}}.$\\

\noindent {\bf Remark 1.} From the \textit{IA 1} above, we can
derive that:

(1)  For $0\leq i \leq t-3$, $I_{2^i}$ has the form

$I_{2^i}= \left\{
\begin{array}{lll}
[1,2]=\hbox{[}1,2^{2i+1}\hbox{]}   &\mbox {if $i=0$};\\
\hbox{[}1+(2+2^{3}+\cdots +2^{2i-1}),2^{2i+1}\hbox{]} &\mbox {if
$1\leq i \leq t-3$}.
 \end{array}
 \right.$

(2)  Generally, for each $s\in [1, 2^{t-3}-1]$, let the 2-adic
expansion of $s$ be
$$s=a_02^0+a_12^1+a_2 2^2+\cdots+a_{t-4}2^{t-4}=(a_0a_1a_2 \cdots
a_{t-4})_{2}.$$

Define $i=i_{s}=\min\{j|a_j=1,0\leq j\leq t-4\}$. We have then
\begin{eqnarray*}
 I_s&=&I_{2^i}+a_{i+1}2^{2i+3}+a_{i+2}2^{2i+5}+\cdots +a_{t-4}2^{2t-7}\\
&=&I_{2^i}+2^{2i+3}(a_{i+1}+a_{i+2}2^2+\cdots+a_{t-4}2^{2(t-i-5)})\\
&=&I_{2^i}+2^{2i+3}\lambda,
 \end{eqnarray*}
$~\hbox{where}~\lambda=\lambda_{s}
=a_{i+1}+a_{i+2}2^2+\cdots+a_{t-4}2^{2(t-i-5)}.$

Notice
 $0 \leq \lambda
   \leq 1+2^2+\cdots+ 2^{2(t-i-5)} <2^{2(t-i-4)}$ for $s\leq 2^{t-3}-1$,
   and $I_{2^{t-3}}=I_{2^{t-3}}+2^{2(t-3)+3}\times \lambda$ with
   $\lambda=0$.  It then follows from (1) and  (2) that for any $s\in [1, 2^{t-3}]$,
     there exists $0 \leq \lambda <2^{2(t-i-4)}$   such that
$$I_s=I_{i,\lambda}=I_{2^i}+2^{2i+3}\lambda,
 ~\hbox{where}~0\leq i \leq t-3.$$

\noindent {\bf Example 1}:  When $t=8$,   $I_{s}$ can be given as
follows:\\
 If $s=8=(00010)_{2}$, then $i=3$ and $\lambda=0$,  $I_{s}=I_{2^3}=I_{2^3}+2^9 \cdot
 0$.\\
If  $s=10=(01010)_{2}$, then $i=1$ and $\lambda=4$,
$I_{s}=I_{2^1}+2^7=I_{2^1}+2^5\cdot 4$.
If  $s=15=(11110)_{2}$, then  $i=0$ and $\lambda=21$,
$I_{s}=I_{2^0}+2^3+2^5+2^7=I_{2^0}+2^3\cdot 21$.

\begin{lemma}\label{lemm3.3}
 Let $\delta_2$ be given as
above. If $x\in [\delta_2+2, \delta_2+2^{2t-4}]$ is odd, then $x$ is
not a coset leader.
 \end{lemma}
\begin{proof}
  Note that $x$ can be given by $x=\delta_2+2l$ with $l \in I^{(t)}=[1,
2^{2t-5}]$.  Let  $I^{(t)}=\bigcup \limits_{s=1}^{2^{t-3}} I_s$ be a
partition of $I^{(t)}$ defined as above. Then $l \in I_s$ for some
$s \in [1, 2^{t-3}]$.
 From    Remark 1, $ I_s=I_{2^i}+
 2^{2i+3}\lambda$,  where  $0\leq i \leq t-3$ is determined by
 $s$  and $0 \leq \lambda<2^{2(t-i-4)}$. Thus, one shall denote
  $x=\delta_2+2l=\delta_2+2(l_0+2^{2i+3}\lambda)$ with
$$l_0\in
I_{2^i}= \left\{
\begin{array}{lll}
[1,2]=[1,2^{2i+1}]   &\mbox {if $i=0$};\\
\hbox{[}1+(2+2^{3}+\cdots +2^{2i-1}),2^{2i+1}\hbox{]} &\mbox {if
$1\leq i \leq t-3$}.
 \end{array}
 \right.$$
 Choose $k=2t-2i-3$. We then derive that
\begin{eqnarray*}
2^kx~ \equiv~ y_{_{x,k}}
&=& 2^k(\delta_2+2l_0+2^{2i+4}\lambda)-(\frac{2^{k-1}-1}{3}+\lambda)n\\
&=& 2^k(\delta_2+2l_0)-(\frac{2^{k-1}-1}{3})(6\delta_2+3)+(2^{2t+1}-n)\lambda\\
&=&  2\delta_2+2^{2(t-i-1)}l_0-2^{2(t-i-2)}+1-\lambda.
\end{eqnarray*}
Now, according to  the value ranges of $i, l_0$ and $\lambda$ above,
when  $k=2t-2i-3$,  one  can show
$\frac{n-\delta_2}{2}<y_{_{x,k}}<\frac{n+\delta_2}{2}$ as follows.

First, we study an upper  bound of $y_{_{x,k}}$:
\begin{eqnarray*} y_{_{x,k}}&=&2\delta_2+2^{2(t-i-1)}l_0-2^{2(t-i-2)}+1-\lambda\\
&\leq& 2\delta_2+2^{2(t-i-1)}l_0-2^{2(t-i-2)}+1\\
 &\leq &2\delta_2+2^{2(t-i-1)}\cdot 2^{2i+1}-2^{2(t-i-2)}+1\\
&= &2\delta_2+2^{2t-1}-2^{2(t-i-2)}+1.
\end{eqnarray*}
Next, we will investigate the lower bounds of $y_{_{x,k}}$:

 If $i=0$, then $l_0\in[1,2]$, we get that \begin{eqnarray*}
y_{_{x,k}}&=&  2\delta_2+2^{2(t-1)}l_0-2^{2(t-2)}+1-\lambda\\
 &>& 2\delta_2+2^{2(t-1)}l_0-2^{2(t-2)}+1-2^{2(t-4)}\\
 &\geq &2\delta_2+2^{2(t-1)}\cdot 1-2^{2(t-2)}+1-2^{2(t-4)}.
\end{eqnarray*}
If $1\leq i \leq t-3$, then $l_0\in[1+(2+2^{3}+\cdots
+2^{2i-1}),2^{2i+1}]$,  we have
\begin{eqnarray*}
 y_{_{x,k}} &=&  2\delta_2+2^{2(t-i-1)}l_0-2^{2(t-i-2)}+1-\lambda\\
 &>& 2\delta_2+2^{2(t-i-1)}l_0-2^{2(t-i-2)}+1-2^{2(t-i-4)}\\
 &\geq &2\delta_2+2^{2(t-i-1)}[1+(2+2^{3}+\cdots +2^{2i-1})]-2^{2(t-i-2)}+1-2^{2(t-i-4)}\\
 &=& \!\!2\delta_2+(2^{2t-3}\!\!+\!\!2^{2t-5}\!\!+\cdots+2^{2t-2i-1})\!\!+2^{2(t-i-1)}\!\!-\!\!2^{2(t-i-2)}+1-2^{2(t-i-4)}\\
  &\geq &\!\!2\delta_2\!\!+(2^{2t-3}\!\!+\!\!\cdots+\!\!2^{2t-2(t-3)-1})\!\!+\!\!2^{2[t-(t-3)-1]}\!\!
  -\!\!2^{2[t-(t-3)-2]}\!\!+\!\!1\!-\!\!2^{2[t-(t-3)-4]}\\
 &= &2\delta_2+2^{2t-3} +2^{2t-5}+\cdots +2^5+2^4-2^2+1-2^{-2}.
\end{eqnarray*}
\begin{eqnarray*}
\hbox{Observe that}~ \frac{n+\delta_2}{2}&=&
\frac{7\delta_2+3}{2}=2\delta_2+2^{2t-1}+1,\\
\frac{n-\delta_2}{2}&=&\frac{5\delta_2+3}{2}=2\delta_2+(2^{2t-3}
+2^{2t-5}+\cdots +2^3+2)+2.
\end{eqnarray*}
Combining the previous bounds of $y_{_{x,k}}$, it is easy to check
$\frac{n-\delta_2}{2}<y_{_{x,k}}<\frac{n+\delta_2}{2}$.

If $\frac{n-\delta_2}{2}<y_{_{x,k}}\leq \frac{n-1}{2}$, then
$n-\delta_2<2y_{_{x,k}}\leq n-1$, one can infer there exists a
$j_{_{x,k}}=n-2y_{_{x,k}}\in C_{y_{_{x,k}}}$ such that $1\leq
j_{_{x,k}}<\delta_2$.

If  $\frac{n+1}{2}\leq y_{_{x,k}}<\frac{n+\delta_2}{2}$, $n+1\leq
2y_{_{x,k}}<n+\delta_2$, one can derive there exists  a $j_{_{x,k}}=
2y_{_{x,k}}-n\equiv 2y_{_{x,k}} \in C_{y_{_{x,k}}}$ achieving $1\leq
j_{_{x,k}}<\delta_2$.

Concluding the previous discussions,  when $x\in [\delta_2+2,
\delta_2 +2^{2t-4}]$,  there exists an integer $j_{_{x,k}}
\in[1,\delta_2)$ such that $j_{_{x,k}}\in C_x$, hence
 $x$ is not a coset leader for  $x\in [\delta_2+2, \delta_2 +2^{2t-4}]$.
 \end{proof}
\begin{lemma}\label{lemm3.4}
 Let $\delta_{1}$  and $\delta_{2}$ be given as
above, if an odd integer $x>\delta_{2}$ and $x\not=\delta_{1}$, then
$x$ is not a coset leader.
\end{lemma}
\begin{proof}   According to the previous discussions,  to attain
the desired conclusion, it suffices to verify there exists some
$k\in[0,m-1]$ such that  $y_{_{x,k}}<x$ or $n-y_{_{x,k}}<x$ for
$x>\delta_2$ except $x=\delta_{1}=\frac{n}{3}$. We split into five
cases:

(1): If $x\in [\frac{n+1}{2},n-1]$, then $x<n<2x$ and
$n-y_{_{x,0}}=n-x<x$.

(2): If $x\in [\frac{n}{3}+1,\frac{n-1}{2}]$,  then $2x<n<3x$ and
$n-y_{_{x,1}}=n-2x<x$.

 (3): If $x\in [\lceil\frac{n}{4} \rceil,\frac{n}{3}-1]$, then $3x<n<4x$ and $y_{_{x,2}}=4x-n<x$.

(4): If $x\in [\lceil \frac{n}{5} \rceil,
\lfloor\frac{n}{4}\rfloor]$, then $4x<n<5x$ and
$n-y_{_{x,2}}=n-4x<x$.

(5): If  $x\in [\lceil \frac{3n}{16} \rceil, \lfloor \frac{n}{5}
\rfloor]$, then $15x<3n<16x$ and  $y_{_{x,4}}=16x-3n<x$.

Combining Lemma \ref{lemm3.3}, we can infer that $x$ is  not a coset
leader for each
 $x\in [\delta_2+2, \delta_2
+2^{2t-4}]\bigcup [\lceil \frac{3n}{16} \rceil, n-1] \setminus
\{\frac{n}{3}\}.$ Notice that $\lceil \frac{3n}{16} \rceil
=\delta_2+\lceil \frac{\delta_2}{8} \rceil<\delta_2+2^{2t-4}$.  It
is easy to know if $x\in [\delta_2+2, n-1] \setminus \{\delta_1\}$,
$x$ is not a coset leader, which is equivalent to the desired
conclusion.
\end{proof}
\begin{theorem}\label{theo3.5}
 Let $\delta_{1}, \delta_{2}, \delta_{3}$,
$\delta_{4}$ and $\delta_{5}$ be  given as Lemma 3.2, then they are
the first, second, third, fourth and fifth largest coset leaders,
respectively.
 \end{theorem}
\begin{proof} According to Lemmas 3.2-3.4, one can naturally infer $\delta_1$ and $\delta_2$
 are the  first and second largest coset leaders, respectively.
 To prove
$\delta_{3}$, $\delta_{4}$ and $\delta_{5}$  are the third, fourth
and fifth largest coset leaders, respectively,
  we need to
show neither $\delta_2-4$ nor $\delta_2-6$  is  a  coset leader.

It is not difficult to infer
\begin{eqnarray*}
2^{2t-3}(\delta_2-4) &\equiv&2^{2t-3}(\delta_2-4)-\frac{2^{2t-4}-1}{3}n\\
               &=&2^{2t-3}+2^{2t-5}+\cdots+2+2^{2t-4}<\delta_2-4,\\
-2^{2t-3}(\delta_2-6) &\equiv&\frac{2^{2t-4}-1}{3}n-2^{2t-3}(\delta_2-6)\\
               &=&\delta_2-2^{2t-2}+2^{2t-4}<\delta_2-6,
               \end{eqnarray*}
 which implies that  $\delta_2-4$,
 $\delta_2-6$ both are   not coset leaders, then the desired conclusion can be derived.
\end{proof}
\begin{lemma}\label{lemm3.6}
 If $1 \leq x \leq  2^{t+2}+7$ or $x\in \{\delta_2,
\delta_3, \delta_4, \delta_5\}$,  then $|C_x|=2m$.
\end{lemma}
\begin{proof} Seeking a contradiction,
suppose $|C_x|=k$  with $k<2m$,  it then follows that  $x(2^{k}-1)
\equiv 0$. Since $k|2m$ and $m$ is odd,    we have  $k=m,
\frac{2m}{3}$ or $k \leq \frac{2m}{5}$.

(1) If $k=m$, then $(2^{k}-1, n)=(2^{m}-1, 2^{m}+1)=1$. From $1 \leq
x \leq 2^{t+2}+7<n$, one has $x(2^{k}-1) \not \equiv 0$.

(2)  If $k=\frac{2m}{3}$, then $m\equiv 0(\bmod~3)$ and $(2^{k}-1,
n=2^{m}+1)=2^{\frac{m}{3}}+1$. Since $1 \leq x \leq
2^{t+2}+7<\frac{n}{2^{\frac{m}{3}}+1}$
$=2^{\frac{2m}{3}}-2^{\frac{m}{3}}+1$,   $x(2^{k}-1) \not \equiv 0
$.

(3)  If $k\leq \frac{2m}{5}$, we then obtain that  $1\leq
x(2^{k}-1)\leq (2^{t+2}+7)(2^{\frac{2m}{5}}-1) <n$, so  $x(2^{k}-1)
\not \equiv 0 $.

Collecting all the above cases, we then conclude that $x(2^{k}-1)
\not \equiv 0 $, a contradiction yields. Hence $|C_x|=2m$ for $1
\leq x \leq 2^{t+2}+7$.

In the proof of Lemma \ref{lemm3.2}, we have derived that when $x\in
\{\delta_2, \delta_3, \delta_4, \delta_5\}$,  $y_{_{x,k}}>x$ for
$1\leq k \leq 2m-1$, obviously,  $x(2^{k}-1) \not \equiv 0 $ and
$|C_x|=2m$.
 \end{proof}
By the previous results, one can get the following theorem.
\begin{theorem}\label{theo3.7} Let  $\delta_1$,
  $\delta_2$, $\delta_3$, $\delta_4$ and $\delta_5$ be
given as Lemma 3.2. If $\delta$ is odd,   then  we have the
following:

 (1) The narrow-sense   BCH codes $\mathcal{C}(n,2,\delta,1)$ have
parameters
 $$ \left\{
\begin{array}{lll}
\hbox{[}n, n-m\delta+5m, d \geq \delta\hbox{]}   &\mbox {if $2^{t+1}+3\leq \delta \leq 2^{t+1}+2^{t}-3$;}\\
\hbox{[}n, n-m\delta+9m, d \geq \delta\hbox{]}   &\mbox {if $2^{t+1}+2^{t}+3\leq \delta\leq 2^{t+2}-9 $;}\\
\hbox{[}n, n-2^{t+2}m+16m, d \geq 2^{t+2}+9\hbox{]}  &\mbox {if $2^{t+2}-7\leq \delta\leq 2^{t+2}+9$;}\\
\hbox{[}n, 2m(i-1)+3, d \geq \delta_i \hbox{]}    &\mbox {if $\delta_{i+1}+2 \leq \delta \leq \delta_{i}(i=1,2,3,4) $;}\\
\hbox{[}n, 1, n\hbox{]}    &\mbox {if $\delta_1+2 \leq \delta \leq
n$.}
 \end{array}
\right.$$

 (2) The   BCH codes $\mathcal{C}(n,2,\delta+1,0)$ have parameters
 $$ \left\{
\begin{array}{lll}
\hbox{[}n, 2^m-m\delta+5m, d \geq 2\delta\hbox{]}   &\mbox {if $2^{t+1}+3\leq \delta \leq 2^{t+1}+2^{t}-3$;}\\
\hbox{[}n, 2^m-m\delta+9m, d \geq 2\delta\hbox{]}  &\mbox {if
$2^{t+1}+2^{t}+3\leq
 \delta\leq 2^{t+2}-9 $;}\\
\hbox{[}n, 2^m-2^{t+2}m+16m, d \geq 2(2^{t+2}+9)\hbox{]}  &\mbox {if  $2^{t+2}-7\leq \delta\leq 2^{t+2}+9$;}\\
\hbox{[}n, 2m(i-1)+2, d \geq 2\delta_i \hbox{]}    &\mbox {if
$\delta_{i+1}+2 \leq \delta \leq \delta_{i}(i=1,2,3,4) $.}
 \end{array}
\right.$$
\end{theorem}

%
%
%
%
%
%
%
%
%
%
%
%
%
%
%
%

\section{  BCH  codes of length $n=2^{m}+1$ with $m=4t+2$}
\label{sec4} In this section,  we fix $m=4t+2 \geq 10$ and
$n=2^{m}+1$.

Similar to  last Section 3, it will be determined   when $x$ is a
 coset leader  for $1\leq x \leq 2^{2t+2}+2^{2t+1}+3$ in Theorem 4.1  in advance.
  Next, we show each $\delta_{i}$ is a coset leader for $1\leq i
\leq 5$ in Lemma 4.2.   For determining the first five largest coset
leaders in Theorem 4.5, we continue to  give  Lemmas 4.3-4.4 on the
basis of Lemma 4.2.  The cardinalities of relevant cyclotomic cosets
will be calculated in Lemma 4.6. Finally,  Theorem 4.7 can be
further given, which lists dimensions of some BCH codes of length
$n$ with designed distance $\delta$ for $2^{2t+1}+3\leq \delta \leq
2^{2t+2}+2^{2t+1}+5$ and $\delta_5+2\leq \delta \leq n$, where
$\delta_5$ is  given as Lemma 4.2.

 The proofs of Theorem 4.1,
 Lemma  4.2 and Theorem 4.7 will be  presented in {\bf Appendixes  D,   E \hbox{and} F},
 respectively.

\begin{theorem}\label{theo4.1}
 Suppose $x$ is odd, then  we have the following:

(1) If $1\leq x\leq 2^{2t+1}-1$, then $x$ is a coset leader, see
Refs. \cite{Ding7} and \cite{Ding8}.

(2) If  $2^{2t+1}+3\leq x\leq 2^{2t+2}-5$, then $x$ is a coset
leader.

(3) If  $2^{2t+2}+5\leq x\leq 2^{2t+2}+2^{2t}-3$, then $x$ is a
coset leader.

(4) If  $2^{2t+2}+2^{2t}+3\leq x\leq 2^{2t+2}+2^{2t+1}-3$, then $x$
is a coset leader.

(5) If $x=2^{2t+1}+1$, $2^{2t+2}-3\leq x \leq 2^{2t+2}+3$,
$2^{2t+2}+2^{2t}-1\leq x \leq 2^{2t+2}+2^{2t}+1$ or
 $2^{2t+2}+2^{2t+1}-1\leq x \leq 2^{2t+2}+2^{2t+1}+3$, then $x$ is not a
coset leader.
\end{theorem}
\begin{lemma}\label{lemm4.2}  Denote $\delta_{1}=\frac{n}{5}$,
$\delta_{2}=2^{4t-1}+\frac{2^{4t}-1}{5}$, $\delta_{3}=\delta_{2}-6$.

If $t=2$, let $\delta_{4}=\delta_{2}-8$ and
$\delta_{5}=\delta_{2}-24$;

 if $t>2$, let $\delta_{4}=\delta_{2}-96$
and $\delta_{5}=\delta_{2}-102$.

 Then $\delta_{1}$, $\delta_{2}$,
$\delta_{3}$, $\delta_{4}$ and $\delta_{5}$ are all coset leaders.
\end{lemma}
In the following, to verify the five largest coset leaders,  we
first introduce an iterative algorithm similar to {\it IA 1} for
partitioning  $ J^{(t)}=[1, 3\cdot2^{4t-6}]$ into $2^{t-2}$
subintervals $ J^{(t)}=J_1\bigcup J_2 \cdots \bigcup J_{2^{t-2}}$
for $t\geq 2$.

\noindent {\bf Iterative algorithm (IA 2):}

Suppose $J_i=[a_i,b_i]$, a  partition of $J^{(t)}=[1,
3\cdot2^{4t-6}]$ is obtained as
 $$J^{(t)}=J^{(t-1)}\bigcup (\bigcup \limits_{u=2^{t-3}+1}^{2^{t-2}} J_{u}),$$
$$\hbox{where} ~J_{u}=\left\{
\begin{array}{lll}
J_{u}=J_{j}+3 \cdot 2^{4\times (t-3)+2} &\mbox {if $u=j+2^{t-3}$, $1\leq j\leq 2^{t-3}-1$;}\\
\mbox{[}a_{2^{t-3}}+3 \cdot2^{4t-10}, 3 \cdot 2^{4t-6}\mbox{]}
&\mbox {if $ u=2^{t-2}$,}
 \end{array}
 \right.$$
with the initial conditions:

$ J_1=[1,3\cdot 2^2]=[1,12]$, $ J_2=[a_1+3\cdot 2^2, 3\cdot
2^6]=[13,192].
 $

For clarity, we discuss the partition in detail as follows.

1) If $t=2$, let $J_1=J_{2^0}=[1,3\cdot 2^{4\times 2-6}]=[1,3\cdot
2^2]=[1,12]=[a_1,b_1]$;

2) If $t=3$, let $J_2=J_{2^1}=[a_1+3\cdot 2^{4\times 0+2}, 3\cdot
2^{4\times 3-6}]$=$[a_1+3\cdot 2^2, 3\cdot 2^6]=[a_2,b_2]$;

3) If $t=4$, $J_3=J_{1+2^1}=J_1+3 \cdot 2^{4\times 1+2}=[a_1+3 \cdot
2^{6}, b_1+3 \cdot 2^{6}]=[a_3,b_3]$,

$J_4=J_{2^2}=[a_2+3 \cdot 2^{4\times 1+2}, 3 \cdot 2^{10}]=[a_2+3
\cdot 2^6, 3 \cdot 2^{10}]=[a_4,b_4]$;

4) Let $t\geq 5$.  Suppose a partition of $ J^{(t-1)}$ is given by

 $J^{(t-1)}=[1, 2^{4t-10}]$ $=J_1\bigcup J_2 \cdots \bigcup
J_{2^{t-3}}=\bigcup \limits_{j=1}^{2^{t-3}} J_j$, where
$J_{j}=[a_j,b_j]$.\\

  For $u=j+2^{t-3}$ with $1\leq j\leq 2^{t-3}-1$, define that

$J_{u}=J_{j}+3 \cdot 2^{4\times (t-3)+2}$ $=J_{j}+3
\cdot2^{4t-10}=[a_j+3 \cdot2^{4t-10},b_j+3 \cdot2^{4t-10}].$

For $u=2^{t-2}$,   define that

$J_{2^{t-2}}=[a_{2^{t-3}}+3 \cdot 2^{4\times (t-3)+2}, 3 \cdot
2^{4\times (t-2)+2}]$=$[a_{2^{t-3}}+3 \cdot2^{4t-10}, 3 \cdot
2^{4t-6}].$

Then, a partition of $J^{(t)}=[1, 3 \cdot 2^{4t-6}]$ is obtained as\\
 $J^{(t)}$ $=\bigcup
\limits_{s=1}^{2^{t-3}}\!\!J_s \bigcup (\!\!\bigcup
\limits_{u=2^{t-3}+1}^{2^{t-2}}\!\!J_{u})$
  $=J_{1}\cup J_{2} \cdots
\cup J_{2^{t-3}}\cup J_{2^{t-3}+1} \cdots \cup J_{2^{t-2}}.$\\

\noindent {\bf Remark 2.} From the \textit{IA 2} above, one can
derive the following:

(1): $J_{2^i}= \left\{
\begin{array}{lll}
\hbox{[}1,12\hbox{]}=\hbox{[}1,3\cdot 2^{4i+2}\hbox{]}   &\mbox {if $i=0$};\\
\hbox{[}1+3 \cdot (2^2+2^6+\cdots+2^{4i-2}),3 \cdot2^{4i+2}\hbox{]}
&\mbox {if $1\leq i \leq t-2$}.
 \end{array}
 \right.$

(2):  Generally, for each $s\in [1, 2^{t-2}-1]$  with its 2-adic
expansion  $s=(a_0a_1a_2 \cdots a_{t-3})_{2},$ denote
$i=i_{s}=\hbox{min}\{j|a_j=1,0\leq j\leq t-3\}$.  Similar to $I_s$
in Section 3,    $J_s$ can be given by
\begin{eqnarray*}
J_s&=&J_{2^i}+a_{i+1}2^{4i+6}+a_{i+2}2^{4i+10}+\cdots +a_{t-3}2^{4t-10}\\
&=&J_{2^i}+2^{4i+6}(a_{i+1}+a_{i+2}2^4+\cdots+a_{t-3}2^{4(t-i-4)})\\
&=&J_{2^i}+2^{4i+6}\lambda,
 \end{eqnarray*}
$~\hbox{where}~\lambda
=a_{i+1}+a_{i+2}2^4+\cdots+a_{t-3}2^{4(t-i-4)}.$

Notice $0 \leq \lambda \leq 1+2^4+\cdots+ 2^{4(t-i-4)}
<2^{4(t-i-3)}$ for $s\leq 2^{t-2}-1$,
   and $J_{2^{t-2}}=J_{2^{t-2}}+2^{4(t-2)+6}\times \lambda$ with
   $\lambda=0$. Thus, from (1) and  (2), one can then derive that for each $s\in [1, 2^{t-2}]$
    there exists $0 \leq \lambda
<2^{4(t-i-3)}$ such that

$$J_s=J_{i,\lambda}=J_{2^i}+2^{4i+6}\lambda,
 ~\hbox{where}~0\leq i \leq t-2.$$

\noindent{\bf Example 2}:  When $t=8$, then $J_{s}$ can be given
as follows:\\
 if $s=8=(00010)_{2}$, then $i=3$ and $\lambda=0$,   $J_{s}=J_{2^3}=J_{2^3}+2^{18} \cdot0$.\\
if  $s=10=(01010)_{2}$, then $i=1$ and $\lambda=2^{4}$,
$J_{s}=J_{2^1}+2^{14}=J_{2^1}+2^{10}\cdot 2^{4}$.\\
if  $s=15=(11110)_{2}$, then $i=0$ and $\lambda=273$,
$J_{s}\!=\!J_{2^0}+2^6+2^{10}+2^{14}\!=\!J_{2^0}+2^6\cdot273$.

\begin{lemma}\label{lemm4.3} Let $\delta_2$ be given as Lemma 4.2. If $x\in
[\delta_2+2, \delta_2+3 \cdot 2^{4t-5}]$ is odd, then $x$ is not a
coset leader.
\end{lemma}
\begin{proof} Similar to the proof of Lemma 3.3,
    one can denote
  $x=\delta_2+2(l_0+2^{4i+6}\lambda)$ with
$ l_0\in J_{2^i}.$

 Choose $k=4t-4i-5$.  We then have that
\begin{eqnarray*}
2^kx\equiv y_{_{x,k}}
&=& 2^k(\delta_2+2l_0+2^{4i+7}\lambda)-(2^{k-3}+\frac{2^{k-2}-2}{5}+\lambda)n\\
&=& 2^k(\delta_2+2l_0)-(\frac{7\cdot 2^{k-3}-2}{5})(\frac{40\delta_2+15}{7})+(2^{4t+2}-n)\lambda\\
&=& \frac{16\delta_2+6}{7}+2^{4(t-i-1)}l_0-3\cdot
2^{4(t-i-2)}-\lambda.
\end{eqnarray*}
Now, according to  the value ranges of $i, l_0$ and $\lambda$ above,
one can show $\frac{n-\delta_2}{2}<y_{_{x,k}}<\frac{n+\delta_2}{2}$
for  $k=4t-4i-5$.

First, we study an upper bound of $y_{_{x,k}}$:
\begin{eqnarray*}
y_{_{x,k}}&=& \frac{16\delta_2+6}{7}+2^{4(t-i-1)}l_0-3\cdot2^{4(t-i-2)}-\lambda\\
 &\leq&
\frac{16\delta_2+6}{7}+2^{4(t-i-1)}l_0-3\cdot
2^{4(t-i-2)}\\
 &\leq & \frac{16\delta_2+6}{7}+2^{4(t-i-1)}(3 \cdot2^{4i+2})-3\cdot
2^{4(t-i-2)}\\
&= &
\frac{16\delta_2+6}{7}+3 \cdot 2^{4t-2}-3 \cdot2^{4(t-i-2)}.
\end{eqnarray*}

Then, we will provide  the lower bounds of $y_{_{x,k}}$:

If $i=0$, then $l_0\in [1,12]$,  one can obtain that
\begin{eqnarray*}
y_{_{x,k}}&=& \frac{16\delta_2+6}{7}+2^{4(t-1)}l_0-3\cdot 2^{4(t-2)}-\lambda\\
 &>& \frac{16\delta_2+6}{7}+2^{4(t-1)}l_0-3\cdot
2^{4(t-2)}-2^{4(t-3)}\\
 &\geq & \frac{16\delta_2+6}{7}+2^{4(t-1)}\cdot 1-3\cdot
2^{4(t-2)}-2^{4(t-3)}.
\end{eqnarray*}
If $1\leq i \leq t-2$, thus $l_0\in [1+3 \cdot
(2^2+2^6+\cdots+2^{4i-2}),3 \cdot2^{4i+2}]$, we get that
\begin{eqnarray*}
 y_{_{x,k}}
 &=& \frac{16\delta_2+6}{7}+2^{4(t-i-1)}l_0-3\cdot 2^{4(t-i-2)}-\lambda\\
 &>&
 \frac{16\delta_2+6}{7}+2^{4(t-i-1)}l_0-3\cdot
2^{4(t-i-2)}-2^{4(t-i-3)}\\
 &\geq &\!\!\!\!\frac{16\delta_2+6}{7}+\!\!2^{4(t-i-1)}[1+3 \cdot (2^2+2^6+\cdots+2^{4i-2})]-\!\!3\cdot
2^{4(t-i-2)}-\!\!2^{4(t-i-3)}\\
 &= &\frac{16\delta_2+6}{7}+3\cdot(2^{4t-6}+2^{4t-10}+\!\!\cdots
 +2^{4t-4i-2})+\!\!13\cdot2^{4(t-i-2)}-\!\!2^{4(t-i-3)}\\
 &>&\!\!\!\!\frac{16\delta_2+6}{7}+\!\!3\cdot(2^{4t-6}
 +\!\!\cdots+\!\!2^{4t-4(t-2)-2})\!\!+\!\!13\cdot2^{4[t-(t-2)-2]}\!\!-\!\!2^{4[t-(t-2)-3]}\\
 &= &
 \frac{16\delta_2+6}{7}+3\cdot(2^{4t-6} +2^{4t-10}+\cdots +2^6)+13-2^{-4}.
\end{eqnarray*}
\begin{eqnarray*}
\hbox{Observe that}~ \frac{n+\delta_2}{2}&=&
\frac{16\delta_2+6}{7}+\frac{15\delta_2+3}{14}=\frac{16\delta_2+6}{7}+3\cdot
2^{4t-2},\\
\frac{n-\delta_2}{2}&=& \frac{16\delta_2+6}{7}+3(2^{4t-6}
+2^{4t-10}+\cdots+2^6)+2^2+1.
\end{eqnarray*}

Combining the previous bounds of $y_{_{x,k}}$, it is easy to check
$\frac{n-\delta_2}{2}<y_{_{x,k}}<\frac{n+\delta_2}{2}$.
 Similar to the proof of Lemma \ref{lemm3.3},
one  can infer that there exists an integer $j_{_{x,k}}
\in[1,\delta_2)$ such that  $j_{_{x,k}}\in C_{y_{_{x,k}}}$ for
$\frac{n-\delta_2}{2}\leq y_{_{x,k}}<\frac{n+\delta_2}{2}$.

Summarizing  the previous discussions,   when $x\in [\delta_2+2,
\delta_2+3 \cdot 2^{4t-5}]$,  there exists an integer $j_{_{x,k}}
\in [1,\delta_2)$ such that $j_{_{x,k}}\in C_x$,  it then follows
that
 $x$ is not a coset leader for all $x\in [\delta_2+2, \delta_2+3 \cdot 2^{4t-5}]$,
 this completes the proof.  \end{proof}
\begin{lemma}\label{lemm4.4} Let $\delta_{1}$  and $\delta_{2}$ be given as
Lemma 4.2.  If an odd integer $x>\delta_{2}$ and $x\not=\delta_{1}$,
then $x$ is not a coset leader.
\end{lemma}
\begin{proof}
  According to previous discussions,  to attain
the desired conclusion, it suffices to verify there exists
$y_{_{x,k}}<x$ or $n-y_{_{x,k}}<x$ for some $k \in [0, m-1]$ and
$x\in [\delta_2+2,n-1] \setminus \{\delta_1=\frac{n}{5}\}$. We give
our discussions in following cases:

(1): If $x\in [\frac{n+1}{2},n-1]$,  then $x<n<2x$ and
$n-y_{_{x,0}}=n-x<x$;

(2): If $x\in [\frac{n+1}{3},\frac{n-1}{2}]$, then $2x<n<3x$  and
  $n-y_{_{x,1}}=n-2x<x$;

(3): If $x\in [\frac{n+3}{4},\frac{n-2}{3}]$, then $3x<n<4x$ and
$y_{_{x,2}}=4x-n<x$;

(4): If $x\in [\frac{n}{5}+1,\frac{n-1}{4}]$, then $4x<n<5x$ and
$n-y_{_{x,2}}=n-4x<x$;

(5): If  $x\in [\frac{3(n-1)}{16}+1,\frac{n}{5}-1]$, then
$15x<3n<16x$ and  $y_{_{x,4}}=16x-3n<x$.

Combining Lemma \ref{lemm4.3}, one can infer that $x$ is   not a
coset leader for all $x\in [\delta_2+2, \delta_2 +3 \cdot
2^{4t-5}]\bigcup [\frac{3(n-1)}{16}+1, n-1] \setminus
\{\frac{n}{5}\}.$  Notice that $\frac{3(n-1)}{16}+1 =\delta_2+
\frac{2^{4t-2}+1}{5}+1 <\delta_2+3 \cdot 2^{4t-5}$,  the conclusion
can be easily derived from the discussions above.\end{proof}
 \begin{theorem}\label{theo4.5}
  Let $\delta_{1}, \delta_{2}, \delta_{3}$,
$\delta_{4}$ and $\delta_{5}$ be  given as Lemma 4.2, then they are
the first, second, third, fourth and fifth largest coset leaders,
respectively.
\end{theorem}
 \begin{proof} According to Lemmas 4.2-4.4, one can naturally infer that $\delta_1$ and $\delta_2$
 are the  first and second largest coset leaders, respectively.
   To verify that  $\delta_3$
is the third  largest coset leader,  it suffices to show neither
$\delta_2-4$ nor  $\delta_2-2$ is  not a coset leader. It is not
difficult to derive that
\begin{eqnarray*}
-2^{4t-1}(\delta_2-4) &\equiv&(2^{4t-4}+\frac{2^{4t-3}-2}{5})n-2^{4t-1}(\delta_2-4)\\
               &=&\delta_2-2^{4t-4}-\frac{2^{4t-2}+1}{5}<\delta_2-4,\\
2^{4t-1}(\delta_2-2) &\equiv&2^{4t-1}(\delta_2-2)-(2^{4t-4}+\frac{2^{4t-3}-2}{5})n\\
               &=&2^{4t-1}-\frac{2^{4t-3}-2}{5}-2^{4t-4}<\delta_2-2,
               \end{eqnarray*}
which implies  that $\delta_2-4$ and  $\delta_2-2$ both are not
coset leaders.

Similarly, it shall be obtained that $x$ is not a coset leader for
$\delta_5 < x <\delta_4$ and $\delta_4 <x< \delta_3$, the detailed
proof is omitted.  As thus, it is easy to know $\delta_4$(resp.
$\delta_5$) is the fourth (resp. fifth) largest coset leader,
 which completes the proof.\end{proof}

\begin{lemma}\label{lemm4.6}  Let $\delta_{1}, \delta_{2}, \delta_{3}$,
$\delta_{4}$ and $\delta_{5}$ be  given as Lemma 4.2. If $1 \leq x
\leq  2^{2t+2}+2^{2t+1}+3$ or $x\in \{\delta_2, \delta_3, \delta_4,
\delta_5\}$,  then $|C_x|=2m$.
\end{lemma}
\begin{proof} Seeking a contradiction, if $|C_x|=k$ with $k<2m$,
 we have $x(2^{k}-1) \equiv 0$. From $k|2m$, one can deduce
 $k=m, \frac{2m}{3},  \frac{2m}{4}$ or $k \leq \frac{2m}{5}$.

(1): If $k=\frac{2m}{2}=m$ or $k=\frac{2m}{4}=\frac{m}{2}$, then
$(2^{k}-1, n)=(2^{m}-1, 2^{m}+1)=(2^{\frac{m}{2}}-1, 2^{m}+1)=1$ and
$1 \leq x \leq 2^{2t+2}+2^{2t+1}+3<n$,  $x(2^{k}-1)\not \equiv 0$.

(2):  If $k=\frac{2m}{3}$, then $m\equiv 0 (\bmod ~3)$ and
$(2^{k}-1, n=2^{m}+1)=2^{\frac{m}{3}}+1$. Since $1 \leq x \leq
2^{2t+2}+2^{2t+1}+3<\frac{n}{2^{\frac{m}{3}}+1}$
$=2^{\frac{2m}{3}}-2^{\frac{m}{3}}+1$,   $x(2^{k}-1) \not \equiv 0$.

(3): If $k\leq \frac{2m}{5}$, then $1\leq x(2^{k}-1)\leq
(2^{2t+2}+2^{2t+1}+3)(2^{k}-1) <n$, $x(2^{k}-1) \not \equiv 0 $.

Collecting  the discussions above, we then conclude that $x(2^{k}-1)
\not \equiv 0 $, a contradiction yields. Hence $|C_x|=2m$ for $1
\leq x \leq 2^{2t+2}+2^{2t+1}+3$.

In the proof of Lemma \ref{lemm4.2}, we have derived that when $x\in
\{\delta_2, \delta_3, \delta_4, \delta_5\}$,
 $y_{_{x,k}}>x$
for $1\leq k \leq 2m-1$, it then follows that $x(2^{k}-1) \not
\equiv 0 $ and $|C_x|=2m$. \end{proof}
 According to  the results
above, one can get the following theorem.
\begin{theorem}\label{theo4.7} Denote that $a=2^{2t+2}+2^{2t+1}$.
  Let $\delta_1$, $\delta_2$, $\delta_3$, $\delta_4$ and $\delta_5$ be
given as Lemma 4.2. If $\delta$ is odd,  then   we have the
following:

 (1) The narrow-sense   BCH codes $\mathcal{C}(n,2,\delta,1)$ have
parameters
 $$ \left\{
\begin{array}{lll}
\hbox{[}n, n-m\delta+3m, d \geq \delta\hbox{]}  &\mbox {if $2^{2t+1}+3\leq \delta\leq 2^{2t+2}-5$;}\\
\hbox{[}n, n-m\delta+11m, d \geq \delta\hbox{]}  &\mbox {if $2^{2t+2}+5\leq \delta\leq 2^{2t+2}+2^{2t}-3$;}\\
\hbox{[}n, n-m\delta+15m, d \geq \delta\hbox{]}  &\mbox {if
$2^{2t+2}+2^{2t}+3\leq \delta \leq a-3$;}\\
\hbox{[}n, n-ma+16m, d \geq a+5 \hbox{]}  &\mbox {if $a-1 \leq \delta\leq a+5$;}\\
\hbox{[}n, 2m(i-1)+5, d \geq \delta_i\hbox{]}    &\mbox {if $\delta_{i+1}+2 \leq \delta \leq \delta_{i}(i=1,2,3,4) $;}\\
\hbox{[}n, 1, n\hbox{]}    &\mbox {if $\delta_1+2 \leq \delta \leq
n$.}
 \end{array}
\right.$$

 (2) The   BCH codes $\mathcal{C}(n,2,\delta+1,0)$ have parameters
 $$ \left\{
\begin{array}{lll}
\hbox{[}n, 2^m-m\delta+3m, d \geq 2\delta\hbox{]}  &\mbox {if $2^{2t+1}+3\leq \delta\leq 2^{2t+2}-5$;}\\
\hbox{[}n, 2^m-m\delta+11m, d \geq 2\delta\hbox{]}  &\mbox {if $2^{2t+2}+5\leq \delta\leq 2^{2t+2}+2^{2t}-3$;}\\
\hbox{[}n, 2^m-m\delta+15m, d \geq 2\delta\hbox{]}  &\mbox {if
$2^{2t+2}+2^{2t}+3\leq\delta \leq a-3$;}\\
\hbox{[}n, 2^m-ma+16m, d \geq 2(a+5)\hbox{]}  &\mbox {if $a-1 \leq \delta\leq a+5$;}\\
\hbox{[}n, 2m(i-1)+4, d \geq 2\delta_i\hbox{]}    &\mbox {if
$\delta_{i}+2 \leq \delta \leq \delta_{i}(i=1,2,3,4)$.}
\end{array}
\right.$$
\end{theorem}

\section{BCH  codes of length $n=2^{m}+1$ with $m=8t+4$}
\label{sec5} In this section, suppose that  $n=2^{m}+1$ with $m=8t+4
\geq 12$,  we only list the main results, omitting their proofs
because of the   similarity with ones in the two sections above.
\begin{theorem}\label{theo5.1}
 Suppose that $x$ is odd, then  we have the following:

(1) If  $1\leq x\leq 2^{4t+2}-1$, then $x$ is a coset leader.

(2) If   $2^{4t+2}+3\leq x\leq 2^{4t+3}-5$, then $x$ is a coset
leader.

(3) If   $2^{4t+3}+5\leq x\leq 2^{4t+3}+2^{4t+1}-3$, then $x$ is a
coset leader.

(4) If   $2^{4t+3}+2^{4t+1}+3\leq x\leq 2^{4t+3}+2^{4t+2}-3$, then
$x$ is a coset leader.

(5) If  $2^{4t+3}+2^{4t+2}+5\leq x\leq
2^{4t+3}+2^{4t+2}+2^{4t+1}-3$, then $x$ is a coset leader.

(6) If $x=2^{4t+2}+1$,
       $2^{4t+3}-3\leq x \leq 2^{4t+3}+3$,
       $2^{4t+3}+2^{4t+1}-1\leq x \leq 2^{4t+3}+2^{4t+1}+1$,
       $2^{4t+3}+2^{4t+2}-1\leq x \leq 2^{4t+3}+2^{4t+2}+3$ or
$2^{4t+3}+2^{4t+2}+2^{4t+1}-1\leq x \leq
2^{4t+3}+2^{4t+2}+2^{4t+1}+1$,
 then $x$ is not a coset leader.
 \end{theorem}
\begin{theorem}\label{theo5.2}   Let $\delta_{1}=\frac{3n}{17}$.

 If $t=1$, let $\delta_{2}=\delta_{1}-6$,
 $\delta_{3}=\delta_{2}-24$, $\delta_{4}=\delta_{3}-2$
and $\delta_{5}=\delta_{4}-38$;

If $t=2$, let $\delta_{2}=\delta_{1}-\frac{\delta_1+45}{128}$ and
$\delta_{3}=\delta_{2}-90$,  $\delta_{4}=\delta_{3}-6$,
$\delta_{5}=\delta_{4}-384$;

If $t\geq 3$, let $\delta_{2}=\delta_{1}-\frac{\delta_1+45}{128}$
and $\delta_{3}=\delta_{2}-90$, $\delta_{4}=\delta_{3}-22950$,
$\delta_{5}=\delta_{4}-90$.

Then they are the first, second, third, fourth and fifth largest
coset leaders, respectively.
 \end{theorem}
\begin{lemma}\label{lemm5.3}
 Let ${\delta_1, \delta_2, \delta_3, \delta_4,
\delta_5}$ be given as Theorem 5.2.  If $1 \leq x \leq
2^{4t+3}+2^{4t+2}+2^{4t+1}+1$ or $x\in \{\delta_2, \delta_3,
\delta_4, \delta_5\}$,  then  $|C_x|=2m$ and $|C_{\delta_1}|=8$.
\end{lemma}
\begin{theorem}\label{theo5.4}
Suppose that  $b=2^{4t+3}+2^{4t+2}+2^{4t+1}$. Let $\delta_1$,
$\delta_2$, $\delta_3$, $\delta_4$ and $\delta_5$ be given as
Theorem 5.2.  If  $\delta$ is odd,  then  we have the following:

 (1) The narrow-sense   BCH codes $\mathcal{C}(n,2,\delta,1)$ have
parameters
 $$ \left\{
\begin{array}{lll}
\hbox{[}n, n-m\delta+3m, d \geq \delta\hbox{]}   &\mbox {if $2^{4t+2}+3\leq \delta\leq 2^{4t+3}-5$;}\\
\hbox{[}n, n-m\delta+11m, d \geq \delta\hbox{]}  &\mbox {if $2^{4t+3}+5\leq \delta\leq  2^{4t+3}+2^{4t+1}-3$;}\\
\hbox{[}n, n-m\delta+15m, d \geq \delta\hbox{]}  &\mbox {if $2^{4t+3}+2^{4t+1}+3\leq \delta \leq 2^{4t+3}+2^{4t+2}-3$;}\\
\hbox{[}n, n-m\delta+21m, d \geq \delta\hbox{]}  &\mbox {if $2^{4t+3}+2^{4t+2}+5\leq \delta \leq b-3$;}\\
\hbox{[}n, n-mb+22m, d \geq b+3\hbox{]}  &\mbox {if $b-1 \leq \delta \leq b+3$;}\\
\hbox{[}n, 2m(i-1)+9, d \geq \delta_i\hbox{]}    &\mbox {if $\delta_{i+1}+2 \leq \delta \leq \delta_i(i=1,2,3,4)$;}\\
\hbox{[}n, 1, n\hbox{]}    &\mbox {if $\delta_1+2 \leq \delta \leq
n$.}
 \end{array}
\right.$$

 (2) The  BCH codes $\mathcal{C}(n,2,\delta+1,0)$ have parameters
 $$ \left\{
\begin{array}{lll}
[n, 2^m-m\delta+3m, d \geq 2\delta]  &\mbox {if $2^{4t+2}+3\leq \delta\leq 2^{4t+3}-5$;}\\
\hbox{[}n, 2^m-m\delta+11m, d \geq 2\delta\hbox{]}  &\mbox {if $2^{4t+3}+5\leq \delta\leq  2^{4t+3}+2^{4t+1}-5$;}\\
\hbox{[}n, 2^m-m\delta+15m, d \geq 2\delta\hbox{]}  &\mbox {if
$2^{4t+3}+2^{4t+1} \!\!+3\leq \delta \leq  \!\!2^{4t+3}+2^{4t+2} \!\!-3$;}\\
\hbox{[}n, 2^m-m\delta+21m, d \geq 2\delta\hbox{]}  &\mbox {if $2^{4t+3}+2^{4t+2}+5\leq \delta \leq b-3$;}\\
\hbox{[}n, 2^m-mb+22m, d \geq 2(b+3)\hbox{]}  &\mbox {if $b-1 \leq \delta \leq b+3$;}\\
\hbox{[}n, 2m(i-1)+8, d \geq 2\delta_i\hbox{]}    &\mbox {if
$\delta_{i+1}+2 \leq \delta \leq \delta_i(i=1,2,3,4)$.}
 \end{array}
\right.$$
 \end{theorem}
\section{Conclusion}
We have  discussed the dimension of some binary  BCH codes of length
$n=2^m+1$ with designed distance $\delta>2^{\lceil \frac{m}{2}
\rceil}$ for  $m=2t+1$,  $m=4t+2$, $m=8t+4$ and $m\geq10$.   The
main contributions of this paper are summarized as follows:

 1) Some new techniques to find  out coset leaders were presented.

 2) The first five largest coset leaders were determined.

3) The dimensions  of some  classes of BCH codes of length $n$ with
designed distance $\delta>2^{\lceil \frac{m}{2} \rceil}$ were
determined. The Bose distances of these codes were also obtained.

Though it is  not easy to study parameters of antiprimitive BCH
codes as pointed out in \cite{Ding5}, fortunately, we have gone one
step further on the basis of Refs. \cite{Ding7} and \cite{Ding8}. It
is worthwhile and expectant  to develop more techniques  and
achievements on antiprimitive  LCD BCH codes over finite fields. We
believe  that our results will shed light on  BCH codes and cyclic
codes. It is also hoped that these results will work to discuss
constructions of LCD codes from negacyclic codes over finite fields
as did in \cite{Shixin}.
\section*{Acknowledgements}

This work is a revised edition  on the basis of its original
edition, which was submitted to the journal "Finite Fields and Their
Applications" on May 4, 2017. Here, we are  greatly indebted to two
anonymous reviewers and the Associate Editor, Prof. Pascale Charpin,
for their  comments and suggestions that much improved the
presentation and quality of this paper.

 The first author would like to express his gratitude to Prof.
Cunsheng Ding for  his helpful suggestions by e-mail communication
and inspiring discussions during National Conference on Coding
theory and Cryptography in Hangzhou, China.

 This work is supported by National Natural Science Foundation
of China  under Grant No.11471011 and Natural Science Foundation of
Shaanxi  under Grant No.2017JQ1032.
\section*{Appendix}

 First of all, we present the following Lemma
 to get the minimum values  in the proofs of Theorems 3.1 and 4.1.

\noindent{\bf Lemma 0.1} Let $f(k)=2^{-k}a+2^{k}b$,  where $a, b$
and $k$ are positive real numbers. If $k_2\geq k_1\geq
\frac{\log_{2}a-\log_{2}b}{2}$, then $f(k_2)\geq f(k_1)$.
\begin{proof} Since $f'(k)=(-2^{-k}a+2^{k}b)\ln 2$, one can
easily  deduce that if $k\geq \frac{\log_{2}a-\log_{2}b}{2}$ , then
$ f'(k)\geq 0$. It    follows that  when $k\geq
\frac{\log_{2}a-\log_{2}b}{2}$, $f(k)$ is monotonically increasing,
this completes the proof.
\end{proof}

\subsection*{Appendix A: The proof of Theorem 3.1}
\begin{proof} To prove this theorem, it suffices to vertify the items (2)-(4).

(2): Since $x$ is odd, $x$ can be denoted by
$$x=2^{t+1}+1+2l, \hbox{where}~l \in I=[1, 2^{t-1}-2].$$

 To verify  (2), it suffices to  prove  $y_{_{x,k}}-x \geq 0$  and
$n-y_{_{x,k}}-x \geq 0$ for $k\in[0,m-1=2t]$.   We first determine
$y_{_{x,k}}$, then show $y_{_{x,k}}-x \geq 0$  and $n-y_{_{x,k}}-x
\geq 0$ according to different $k$.

 (2.1): When $k=0,1,2,\cdots,t-1$, from $2^{t+1}+3\leq x\leq 2^{t+1}+2^{t}-3$,
 we have $x\leq 2^kx <n$, hence, $y_{_{x,k}}=2^kx\geq
x$, and
\begin{eqnarray*}
 n-y_{_{x,k}}-x&=&2^{2t+1}+1-(2^k+1)x\\
 &\geq&2^{2t+1}+1-(2^{t-1}+1)x\\
 &\geq&2^{2t+1}+1-(2^{t-1}+1)(2^{t+1}+2^{t}-3)\\
 &= & 2^{2t-1}-2^{t-1}-2^t+4>0.
\end{eqnarray*}
(2.2): When $k=t, t+1$, we have $y_{_{x,k}}=2^kx-2^{k-t}n$, thus

\begin{eqnarray*}
y_{_{x,k}}-x &=&(2^k-1)x-2^{k-t}n\\
&=& (2^k-1)(2^{t+1}+1+2l)-2^{k-t}(2^{2t+1}+1)\\
&=&(2^{k+1}-2)l+2^k-2^{t+1}-1-2^{k-t}\\
&\geq&(2^{k+1}-2)\cdot1+2^k-2^{t+1}-1-2^{k-t}\\
 &=&(3-2^{-t})\cdot2^k-2^{t+1}-3\\
  &\geq&(3-2^{-t})\cdot2^t-2^{t+1}-3= 2^{t}-4>0,
  \end{eqnarray*}
  \begin{eqnarray*}
   n-y_{_{x,k}}-x &=&(2^{k-t}+1)n-(2^k+1)x\\
&=& (2^{k-t}+1)(2^{2t+1}+1)-(2^k+1)(2^{t+1}+1+2l)\\
&=&  2^{k-t}+2^{2t+1}-2^{t+1}-2^{k}-(2^{k+1}+2)l\\
&\geq& 2^{k-t}+2^{2t+1}-2^{t+1}-2^{k}-(2^{k+1}+2)(2^{t-1}-2)\\
&=& 2^{2t+1}-2^{t}-2^{t+1}+4-(2^t-2^{-t}-3)2^{k}\\
&\geq& 2^{2t+1}-2^{t}-2^{t+1}+4-(2^t-2^{-t}-3)2^{t+1}\\
&=&3(2^{t}+2)>0.
\end{eqnarray*}

(2.3): When $k=t+2,t+3,\cdots,2t-1$, it is a little complex to
determine   $y_{_{x,k}}$ and    check   $y_{_{x,k}}-x \geq 0$ along
with $n-y_{_{x,k}}-x \geq 0$. To achieve this, for each $k$,  we
divide  the value range $I=[1, 2^{t-1}-2]$ of $l$ into $2^{k-1-t}$
 subintervals as follows:

 $I_{_{\lambda,k}}=[1,\lambda 2^{2t-k}-1]~\hbox{for}~ \lambda =1$,

 $I_{_{\lambda,k}}=[(\lambda-1)2^{2t-k},
  \lambda2^{2t-k}-1] ~\hbox{for}~ \lambda \in [2, 2^{k-t-1}-1]$,

 $I_{_{\lambda,k}}=[(\lambda-1)2^{2t-k},2^{t-1}-2] ~\hbox{for}~ \lambda = 2^{k-t-1}$.

Fix $k$, for each  $\lambda \in [1, 2^{k-t-1}]$,  $\lambda$ is
called the {\it identity tag} of the  subinterval
$I_{_{\lambda,k}}=[l_{_{\lambda,b}},l_{_{\lambda,e}}]$. Thus, for
$x=2^{t+1}+1+2l$ and given $\lambda$,   it is not difficult to
derive $y_{_{x,k}}=2^kx-(2^{k-t}+\lambda-1)n$ if  $l\in
I_{_{\lambda,k}}$.

 {\bf Case 2.3.1:} Firstly, we show $y_{_{x,k}}-x\geq0$.
 For  general $\lambda$,
\begin{eqnarray*}
y_{_{x,k}}-x &=&(2^k-1)x-(2^{k-t}+\lambda-1)n\\
&=& (2^k-1)(2^{t+1}+1+2l)-(2^{k-t}+\lambda-1)(2^{2t+1}+1)\\
&=&(2^k-1)(1+2l)-2^{t+1}-2^{k-t}-(\lambda-1)(2^{2t+1}+1).
 \end{eqnarray*}
When $\lambda=1$, from $l\geq 1$ and $k\geq t+2$, we have
\begin{eqnarray*}
y_{_{x,k}}-x&=&(2^k-1)(1+2l)-2^{t+1}-2^{k-t}\\
&\geq& (2^k-1)(1+2\cdot1)-2^{t+1}-2^{k-t}\\
&=& (3-2^{-t})\cdot 2^k-2^{t+1}-3\\
&\geq& (3-2^{-t})\cdot 2^{t+2}-2^{t+1}-3\\
             &=&5\cdot 2^{t+2}-7>0.
            \end{eqnarray*}
For given  $\lambda \in [2, 2^{k-t-1}]$, from $l\in
I_{_{\lambda,k}}=[l_{_{\lambda,b}}=(\lambda-1)2^{2t-k},l_{_{\lambda,e}}]$,
we have
\begin{eqnarray*}
y_{_{x,k}}-x
&=&(2^k-1)(1+2l)-2^{t+1}-2^{k-t}-(\lambda-1)(2^{2t+1}+1)\\
&\geq&(2^k-1)(1+2\cdot
l_{_{\lambda,b}})-2^{t+1}-2^{k-t}-(\lambda-1)(2^{2t+1}+1)\\
&=& 2^{2t+1-k}+2^k-2^{t+1}-2^{k-t}-(2^{2t-k+1}+1)\lambda\\
&\geq& 2^{2t+1-k}+2^k-2^{t+1}-2^{k-t}-(2^{2t-k+1}+1)\cdot2^{k-t-1}\\
&=& 2^{2t+1-k}+2^k(1-2^{-t}-2^{-t-1})-2^{t+1}-2^{t}\\
&\geq& \!2^{2t+1-(t+2)}+\!\!2^{t+2}(1-2^{-t}-2^{-t-1})\!\!-2^{t+1}-\!2^{t}(\hbox{see Lemma 0.1}) \\
&=& 2^{t}+2^{t-1}-6>0.
\end{eqnarray*}

As thus, we have shown $y_{_{x,k}}-x\geq0$ for
$k=t+2,t+3,\cdots,2t-1$.

{\bf Case 2.3.2:} Secondly,  we will show $n-y_{_{x,k}}-x\geq0$.
For  general $\lambda$,
\begin{eqnarray*}
n-y_{_{x,k}}-x 
&=& (2^{k-t}+\lambda)(2^{2t+1}+1)-(2^k+1)(2^{t+1}+1+2l).
\end{eqnarray*}
If $\lambda$ satisfies   $1\leq\lambda \leq 2^{k-t-1}-1$ and  $l\in
I_{_{\lambda,k}}=[l_{_{\lambda,b}},l_{_{\lambda,e}}=\lambda2^{2t-k}-1]$,
then
\begin{eqnarray*}
&&n-y_{_{x,k}}-x\\
&\geq& (2^{k-t}+\lambda)(2^{2t+1}+1)-(2^k+1)(2^{t+1}+1+2\cdot l_{_{\lambda,e}})\\
&=&2^{k-t}-2^{t+1}+2^{k}+1-(2^{2t+1-k}-1)\lambda \\
&\geq&2^{k-t}-2^{t+1}+2^{k}+1-(2^{2t+1-k}-1)\cdot(2^{k-t-1}-1) \\
&=& 2^{2t+1-k}+2^k(1+2^{-t}+2^{-t-1})-2^{t+1}-2^{t}\\
&\geq& 2^{2t+1-(t+2)}+2^{t+2}(1+2^{-t}+2^{-t-1})-2^{t+1}-2^{t}(\hbox{see Lemma 0.1})\\
&=& 3(2^{t-1}+2)>0.
\end{eqnarray*}
If  $\lambda= 2^{k-t-1}$, let $l\in
I_{_{\lambda,k}}=[l_{_{\lambda,b}}=(\lambda-1)2^{2t-k},
l_{_{\lambda,e}}=2^{t-1}-2]$,  we get then
\begin{eqnarray*}
n-y_{_{x,k}}-x
&=&(2^{k-t}+2^{k-t-1})(2^{2t+1}+1)-(2^k+1)(2^{t+1}+1+2l)\\
&\geq&(2^{k-t}+\lambda)(2^{2t+1}+1)-(2^k+1)(2^{t+1}+1+2\cdot l_{_{\lambda,e}})\\
&=&2^k(3+2^{-t}+2^{-t-1})-2^{t+1}-2^{t}+3\\
&\geq&2^{t+2}(3+2^{-t}+2^{-t-1})-2^{t+1}-2^{t}+3\\
&=& 9(2^{t}+1)>0.
\end{eqnarray*}
(2.4): When $k=2t$, for $x=2^{t+1}+1+2l$ with $l\in
I=[1,2^{t-1}-2]$,  we have
\begin{eqnarray*}
y_{_{x,k}}&=&2^kx-(2^{k-t}+l)n=2^{2t}-2^{t}-l.
\end{eqnarray*}
 Thus, we check that \begin{eqnarray*}
y_{_{x,k}}-x
&=&2^{2t}-2^{t}-l-(2^{t+1}+1+2l)\\
&=&2^{2t}-2^{t+1}-2^t-1-3l\\
&\geq&2^{2t}-2^{t+1}-2^t-1-3\cdot(2^{t-1}-2)\\
 &=& 2^{2t}-2^{t+2}-2^{t-1}+5>0,\\
n-y_{_{x,k}}-x
&=&2^{2t}-2^{t+1}+2^t-l\\
&\geq&2^{2t}-2^{t+1}+2^t-(2^{t-1}-2)\\
 &=& 2^{2t}-2^{t+1}+2^{t-1}+2>0.
\end{eqnarray*}

Summarizing the four cases above, we then  conclude that
$y_{_{x,k}}-x\geq0$ and $n-y_{_{x,k}}-x\geq0$ for $2^{t+1}+3\leq
x\leq 2^{t+1}+2^{t}-3$
  and   $k\in [0,m-1=2t]$,  (2) follows.\\

(3) For an odd $x \in [2^{t+1}+2^{t}+3,2^{t+2}-9]$,   let
$x=2^{t+1}+2^{t}+1+2l$ with $l \in J=[1, 2^{t-1}-5]$. To verify (3)
holds, we will first determine $y_{_{x,k}}$ and then show
$y_{_{x,k}}-x \geq 0$ and $n-y_{_{x,k}}-x \geq 0$ for all
$k\in[0,m-1=2t]$. Similar  to (2) above, we split into following
cases according to different $k$.

(3.1)  When  $k=0,1,2,\cdots,t-1$, it is clear that $x\leq 2^kx<n$,
then we have $y_{_{x,k}}=2^kx\geq x$ and
\begin{eqnarray*}
 n-y_{_{x,k}}-x&=&2^{2t+1}+1-(2^k+1)x\\
 &\geq&2^{2t+1}+1-(2^k+1)(2^{t+2}-9)\\
 &= & 2^{t-1}+10>0.
\end{eqnarray*}
(3.2): When $k=t$, we have $y_{_{x,k}}=2^kx-n$, hence
\begin{eqnarray*}
y_{_{x,k}}-x &=& (2^t-1)x-n\\
&\geq& (2^t-1)(2^{t+1}+2^{t}+3)-n\\
&=&2^{2t}-4>0,\\
 & &\\
n-y_{_{x,k}}-x &=&2n-(2^t+1)x\\
&\geq&2n-(2^t+1)(2^{t+2}-9)\\
&=&5\cdot2^t+11>0.
\end{eqnarray*}

(3.3): When  $k=t+1$, we have  $y_{_{x,k}}=2^kx-3n$,
 it follows that
\begin{eqnarray*}
y_{_{x,k}}-x &=& (2^{t+1}-1)x-3n\\
&\geq& (2^{t+1}-1)(2^{t+1}+2^{t}+3)-3n\\
&=&3(2^{t}-2)>0,\\
 & &\\
n-y_{_{x,k}}-x &=&4n-(2^{t+1}+1)x\\
 &\geq&4n-(2^{t+1}+1)(2^{t+2}-9)\\
&=&7 \cdot 2^{t+1}+13>0.
\end{eqnarray*}

(3.4): For each $k=t+2,t+3,\cdots, 2t-3$, to determine $y_{_{x,k}}$
and show $y_{_{x,k}}-x\geq 0$ and $n-y_{_{x,k}}-x\geq0$, we divide
the value range $J=[1, 2^{t-1}-5]$ of $l$ into $2^{k-1-t}$
subintervals as follows:

 $J_{_{\lambda,k}}=[1,\lambda2^{2t-k}-1]~\hbox{for}~ \lambda=1$,

 $J_{_{\lambda,k}}=[(\lambda-1)2^{2t-k},
 \lambda2^{2t-k}-1] ~\hbox{for}~ \lambda \in [2,2^{k-t-1}-1]$,

 $J_{_{\lambda,k}}=[(\lambda-1)2^{2t-k},2^{t}-5]~\hbox{for}~ \lambda=2^{k-t-1}$.

For given $k$,  we  can define $\lambda \in [1, 2^{k-t-1}]$ as the
{\it identity tag}  of the subinterval
$J_{_{\lambda,k}}=[l_{_{\lambda,b}},l_{_{\lambda,e}}]$. Fix
$\lambda$, if $x=2^{t+1}+2^{t}+1+2l$ with $l\in J_{_{\lambda,k}}$,
it follows that $$y_{_{x,k}}=2^kx-(2^{k-t}+2^{k-t-1}+\lambda-1)n.$$
As so, we can further verify $y_{_{x,k}}-x>0$ and
$n-y_{_{x,k}}-x>0$.

  {\bf Case 3.4.1:} Firstly, we show $y_{_{x,k}}-x>0$
\begin{eqnarray*}
y_{_{x,k}}-x &=&(2^k-1)x-(2^{k-t}+2^{k-t-1}+\lambda-1)n\\
&=&(2^k-1)(2^{t+1}+2^{t}+1+2l)-(2^{k-t}+2^{k-t-1}+\lambda-1)(2^{2t+1}+1)\\
&=&(2^k-1)(1+2l)-2^{t+1}-2^{t}-2^{k-t}-2^{k-t-1}-(\lambda-1)(2^{2t+1}+1).
 \end{eqnarray*}
When $\lambda=1$, notice that $l\in J_{_{\lambda,k}}=[1,2^{2t-k}-1]$
and $k\geq t+2$, we get that
\begin{eqnarray*}
y_{_{x,k}}-x&=&(2^k-1)(1+2l)-2^{t+1}-2^{t}-2^{k-t}-2^{k-t-1}\\
&\geq&(2^k-1)(1+2\cdot 1)-2^{t+1}-2^{t}-2^{k-t}-2^{k-t-1}\\
            &=& (3-2^{-t}-2^{-t-1})\cdot 2^k-2^{t+1}-2^{t}-3\\
            &\geq& (3-2^{-t}-2^{-t-1})\cdot 2^{t+2}-2^{t+1}-2^{t}-3\\
            &=&3\cdot 2^{t+2}-2^{t+1}-2^{t}-9>0.
\end{eqnarray*}
For given  $\lambda \in [2, 2^{k-t-1}]$, when $l\in
J_{_{\lambda,k}}=[l_{_{\lambda,b}}=(\lambda-1)2^{2t-k},l_{_{\lambda,e}}]$,
we have a similar derivation process:
\begin{eqnarray*}
&&y_{_{x,k}}-x\\
&=&(2^k-1)(1+2l)-2^{t+1}-2^{t}-2^{k-t}-2^{k-t-1}-(\lambda-1)(2^{2t+1}+1)\\
&\geq&(2^k-1)(1+2\cdot l_{_{\lambda,b}})-2^{t+1}-2^{t}-2^{k-t}-2^{k-t-1}-(\lambda-1)(2^{2t+1}+1)\\
&=& 2^{2t+1-k}+2^k-2^{t+1}-2^{t}-2^{k-t}-2^{k-t-1}-(2^{2t-k+1}+1)\lambda\\
&\geq& 2^{2t+1-k}+2^k-2^{t+1}-2^{t}-2^{k-t}-2^{k-t-1}-(2^{2t-k+1}+1)\cdot2^{k-t-1}\\
&=& 2^{2t+1-k}+2^k(1-2^{1-t})-2^{t+2}\\
&\geq& 2^{2t+1-(t+2)}+2^{t+2}(1-2^{1-t})-2^{t+2}(\hbox{see Lemma 0.1})\\
&=& 2^{t-1}-8>0.
\end{eqnarray*}

{\bf Case 3.4.2:} Secondly, we show $n-y_{_{x,k}}-x>0$
\begin{eqnarray*}
n-y_{_{x,k}}-x &=&(2^{k-t}+2^{k-t-1}+\lambda)n-(2^k+1)x\\
&=&
(2^{k-t}+2^{k-t-1}+\lambda)(2^{2t+1}+1)-(2^k+1)(2^{t+1}+2^{t}+1+2l).
\end{eqnarray*}
For given  $\lambda \in [1, 2^{k-t-1}-1]$, let $l\in
J_{_{\lambda,k}}=[l_{_{\lambda,b}},l_{_{\lambda,e}}=\lambda2^{2t-k}-1]$.
Then one can check that
\begin{eqnarray*}
&&n-y_{_{x,k}}-x\\
&\geq&(2^{k-t}+2^{k-t-1}+\lambda)(2^{2t+1}+1)-(2^k+1)(2^{t+1}+2^{t}+1+2\cdot
l_{_{\lambda,e}})\\
&=&2^{k-t}+2^{k-t-1}-2^{t+1}-2^{t}+2^{k}+1-(2^{2t+1-k}-1)\lambda\\
&\geq&2^{k-t}+2^{k-t-1}-2^{t+1}-2^{t}+2^{k}+1-(2^{2t+1-k}-1)(2^{k-t-1}-1)\\
&=& 2^{2t+1-k}+2^k(1+2^{1-t})-2^{t+2}\\
&\geq& 2^{2t+1-(t+2)}+2^{t+2}(1+2^{1-t})-2^{t+2}(\hbox{see Lemma 0.1})\\
&=& 2^{t-1}+8>0.
\end{eqnarray*}
When $\lambda= 2^{k-t-1}$, since $l\in
I_{_{\lambda,k}}=[l_{_{\lambda,b}}=(\lambda-1)2^{2t-k},l_{_{\lambda,e}}=2^{t}-5]$,
it is easy to deduce that
\begin{eqnarray*}
& &n-y_{_{x,k}}-x\\
&=&(2^{k-t}+2^{k-t-1}+2^{k-t-1})(2^{2t+1}+1)-(2^k+1)(2^{t+1}+2^{t}+1+2l)\\
&\geq&(2^{k-t}+2^{k-t-1}+2^{k-t-1})(2^{2t+1}+1)-(2^k+1)(2^{t+1}+2^{t}+1+2l_{_{\lambda,e}})\\
&=&2^k(2^{1-t}+9)-2^{t+2}+9\\
&\geq&2^{t+2}(2^{1-t}+9)-2^{t+2}+9\\
&=& 2^{t+5}+17>0.
\end{eqnarray*}

(3.5): For each  $k=2t-2, 2t-1$, similar to (3.3) and (3.4), we
divide the value range $J=[1, 2^{t-1}-5]$ of $l$ into
$2^{k-1-t}-2^{k+2-2t}$ subintervals as follows:

 $J_{_{\lambda,k}}=[1,\lambda2^{2t-k}-1]~\hbox{for}~ \lambda=1$,

 $J_{_{\lambda,k}}=[(\lambda-1)2^{2t-k},
 \lambda2^{2t-k}-1] ~\hbox{for}~ \lambda \in [2,2^{k-1-t}-2^{k+2-2t}]$.

Let $\lambda \in [1,2^{k-1-t}-2^{k+2-2t}]$,  we  can define
$\lambda$ as the {\it identity tag} of the subinterval
$J_{_{\lambda,k}}=[l_{_{\lambda,b}},l_{_{\lambda,e}}]$.

For $x=2^{t+1}+2^{t}+1+2l$ with $l\in J_{_{\lambda,k}}$, one can
  check that $$y_{_{x,k}}=2^kx-(2^{k-t}+2^{k-t-1}+\lambda-1)n,$$
then we split into following two cases to deduce the desired result.

{\bf Case 3.5.1:} Firstly, we show $y_{_{x,k}}-x>0$
\begin{eqnarray*}
y_{_{x,k}}-x &=&(2^k-1)x-(2^{k-t}+2^{k-t-1}+\lambda-1)n\\
&=&(2^k-1)(2^{t+1}+2^{t}+1+2l)-(2^{k-t}+2^{k-t-1}+\lambda-1)(2^{2t+1}+1)\\
&=&(2^k-1)(1+2l)-2^{t+1}-2^{t}-2^{k-t}-2^{k-t-1}-(\lambda-1)(2^{2t+1}+1).
 \end{eqnarray*}
If $\lambda=1$,  it then  follows from $l\geq1$ and $k\geq2t-2$ that
\begin{eqnarray*}
y_{_{x,k}}-x &=&(2^k-1)(1+2l)-2^{t+1}-2^{t}-2^{k-t}-2^{k-t-1}\\
&\geq&(2^k-1)(1+2\cdot1)-2^{t+1}-2^{t}-2^{k-t}-2^{k-t-1}\\
&=&(3-2^{-t}-2^{-t-1})\cdot 2^k-2^{t+1}-2^{t}-3\\
&\geq&(3-2^{-t}-2^{-t-1})\cdot\!\! 2^{2t-2}-\!\! 2^{t+1}-2^{t}-3(\hbox{see Lemma 0.1})\\
&=&3\cdot 2^{2t-2}-2^{t+1}-2^{t}-2^{t-2}-2^{t-3}-3>0.
\end{eqnarray*}
For given  $\lambda \in [2, 2^{k-t-1}-2^{k+2-2t}]$, notice that
$l\in
J_{_{\lambda,k}}=[l_{_{\lambda,b}}=(\lambda-1)2^{2t-k},l_{_{\lambda,e}}]$,
thus we have
\begin{eqnarray*}
& &y_{_{x,k}}-x\\
&=&(2^k-1)(1+2l)-2^{t+1}-2^{t}-2^{k-t}-2^{k-t-1}-(\lambda-1)(2^{2t+1}+1)\\
&\geq&(2^k-1)(1+2\cdot l_{_{\lambda,b}})-2^{t+1}-2^{t}-2^{k-t}-2^{k-t-1}-(\lambda-1)(2^{2t+1}+1)\\
&=& 2^{2t+1-k}+2^k-2^{t+1}-2^{t}-2^{k-t}-2^{k-t-1}-(2^{2t-k+1}+1)\lambda\\
&\geq& \!\! 2^{2t+1-k}+2^k-2^{t+1}-2^{t}-2^{k-t}\!-\!2^{k-t-1}\!-\!(2^{2t-k+1}+1)(2^{k-t-1}-2^{k+2-2t})\\
&=& 2^{2t+1-k}+2^k(1+ 2^{2-2t}-2^{1-t})-2^{t+2}+8\\
&\geq& 2^{2t+1- (2t-2)}\!\!+\!\! 2^{2t-2}(1+ 2^{2-2t}-2^{1-t})\!\! -2^{t+2}+8(\hbox{see Lemma 0.1})\\
&=& 2^{2t-2}-2^{t+2}-2^{t-1}+17>0.
\end{eqnarray*}
{\bf Case 3.5.2:} Secondly, we show $n-y_{_{x,k}}-x>0$
\begin{eqnarray*}
n-y_{_{x,k}}-x &=&(2^{k-t}+2^{k-t-1}+\lambda)n-(2^k+1)x\\
&=&(2^{k-t}+2^{k-t-1}+\lambda)(2^{2t+1}+1)\!\!-(2^k+1)(2^{t+1}+2^{t}+1+2l).
\end{eqnarray*}
For given $\lambda \in [1, 2^{k-t-1}-2^{k+2-2t}]$, since $l\in
J_{_{\lambda,k}}=[l_{_{\lambda,b}},l_{_{\lambda,e}}=\lambda2^{2t-k}-1]$,
thus one can deduce that
\begin{eqnarray*}
&&n-y_{_{x,k}}-x\\
&\geq&(2^{k-t}+2^{k-t-1}+\lambda)(2^{2t+1}+1)-(2^k+1)(2^{t+1}+2^{t}+1+2l_{_{\lambda,e}})\\
&=&2^{k-t}+2^{k-t-1}-2^{t+1}-2^{t}+2^{k}+1-(2^{2t+1-k}-1)\lambda\\
&\geq&2^{k-t}+2^{k-t-1}-2^{t+1}-2^{t}+2^{k}+1-(2^{2t+1-k}-1)(2^{k-t-1}-2^{k+2-2t})\\
&=& 2^k(1+2^{1-t}-2^{2-2t})-2^{t+2}+9\\
&\geq& 2^{2t-2}(1+2^{1-t}-2^{2-2t})-2^{t+2}+9\\
&=& 2^{2t-2}+2^{t-1}-2^{t+2}+8>0.
\end{eqnarray*}

(3.6): When $k=2t$, according to  $x=2^{t+1}+2^{t}+1+2l$ with $l\in
J=[1, 2^{t-1}-5]$,  we have
\begin{eqnarray*}
y_{_{x,k}}&=&2^kx-(2^{k-t}+l)n\\
&=&2^{2t}(2^{t+1}+2^{t}+1+2l)-(2^{t}+l)(2^{2t+1}+1)\\
&=&2^{2t}-2^{t}-2^{t-1}-l,
\end{eqnarray*}
thus, one can easily derive  from $l\in J=[1, 2^{t-1}-5]$ that
\begin{eqnarray*}
y_{_{x,k}}-x&=&2^{2t}-2^{t}-2^{t-1}-l-(2^{t+1}+2^{t}+1+2l)\\
&=&2^{2t}-2^{t+2}-2^{t-1}-1-3l\\
&\geq&2^{2t}-2^{t+2}-2^{t-1}-1-3(2^{t-1}-5)\\
 &=& 2^{2t}-2^{t+2}-2^{t+1}+14>0,\\
  & &\\
n-y_{_{x,k}}-x&=&(2^{2t+1}+1)-(2^{2t}-2^{t}-2^{t-1}-l)-(2^{t+1}+2^{t}+1+2l)\\
&=&2^{2t}-2^{t+1}+2^{t-1}-l\\
&\geq&2^{2t}-2^{t+1}+2^{t-1}-(2^{t-1}-5)\\
 &=& 2^{2t}-2^{t+1}+5>0.
\end{eqnarray*}

Concluding the previous six cases (3.1)-(3.6), we then conclude that
$y_{_{x,k}}-x\geq0$ and $n-y_{_{x,k}}-x\geq0$ for
$2^{t+1}+2^{t}+3\leq x\leq 2^{t+2}-9$ and each $k\in [0,m-1=2t]$,
which implies that (3) holds.\\

 (4)  It is easy to check that:

$(2^{t+1}-1)2^{3t+1}=(2^{t+1}-1)2^{t}\cdot
2^{2t+1}=(n-1-2^t)(n-1)\equiv 2^{t}+1$;

$(2^{t+1}+1)2^{t}=2^{m}+2^{t}\equiv 2^{t}-1$;

$(2^{t+1}+2^{t}-1)2^{3t+2}=(3n-3-2^{t+1})(n-1)\equiv 2^{t+1}+3$;

$(2^{t+1}+2^{t}+1)2^{t+1}=2^{m+1}+2^{m}+2^{t+1}\equiv 2^{t+1}-3$;

If $i=1,3,5 \hbox{~or~}7$, then

$(2^{t+2}-i)2^{3t}=(2^{t+2}-i)2^{t-1}\cdot 2^{2t+1}=(n-1-i\cdot
2^{t-1})(n-1)\equiv i \cdot 2^{t-1}+1$;

$ (2^{t+2}+i)2^{t-1} =2^{2t+1}+i\cdot2^{t-1}\equiv i\cdot
2^{t-1}-1$.

Combining the definition of a cyclotomic coset, the above congruence
expressions  imply that there exists an odd integer $y\in[1,x-1]$
satisfying $y\in C_{x}$ for each $x$ in (4), hence $x$ is not a
coset leader, (4) follows.\end{proof}

\subsection*{Appendix B: The proof of Lemma 3.2}
\begin{proof} From $\delta_{1}=\frac{n}{3}=2^{2t}-2^{2t-1}+\cdots+4-2+1$,
$\delta_{2}=\frac{n-3}{6}=2^{2t-2}+2^{2t-4}+\cdots+4+1$, we can
obviously know $\delta_{1}, \delta_{2}, \delta_{3}, \delta_{4},
\delta_{5}$ are odd integers. It is easy to derive
 $C_{\delta_{1}}=\{\delta_{1}, 2\delta_{1}\}$,
 which implies that $|C_{\delta_{1}}|=2$ and  $\delta_{1}$ is a coset leader.
We then show that $\delta_{2}$, $\delta_{3}$, $\delta_{4}$ and
$\delta_{5}$ are all also coset leaders.

{\it Step 1:} We show $y_{_{\delta_2,k}}-\delta_2\geq 0$ and
$n-y_{_{\delta_2,k}}-\delta_2\geq 0$ in three cases:

(1.1): If $k=0,1,2$, it is clear that

$y_{_{\delta_2,k}}=2^k\delta_2\geq \delta_2$,
$n-y_{_{\delta_2,k}}-\delta_2=n-2^k\delta_2-\delta_2=(5-2^k)\delta_2+3
>0$.

(1.2):  If $k=3,5,\cdots,2t-1$,   we then get that
\begin{eqnarray*}
y_{_{\delta_2,k}}&=&2^k\delta_2-\frac{2^{k-1}-1}{3}n=2\delta_2-2^{k-1}+1,\\
y_{_{\delta_2,k}}-\delta_2&=&\delta_2-2^{k-1}+1
\geq\delta_2-2^{(2t-1)-1}+1
 =\delta_2-2^{2t-2}+1>0,\\
n-y_{_{\delta_2,k}}-\delta_2&=&3\delta_2+2^{k-1}+2>0.
\end{eqnarray*}

(1.3):  If $k=4,6,\cdots,2t$,   we have then
\begin{eqnarray*}
y_{_{\delta_2,k}}&=&2^k\delta_2-\frac{2^{k-1}-2}{3}n=4\delta_2-2^{k-1}+2,\\
y_{_{\delta_2,k}}-\delta_2&=&3\delta_2-2^{k-1}+2
\geq3\delta_2-2^{2t-1}+2>0,\\
n-y_{_{\delta_2,k}}-\delta_2&=&\delta_2+2^{k-1}+1>0.
\end{eqnarray*}

From the three cases above, we then can conclude $\delta_2$ is a
coset leader.

{\it Step 2:} Now, we show $y_{_{\delta_3,k}}-\delta_3\geq 0$ and
$n-y_{_{\delta_3,k}}-\delta_3\geq 0$.

(2.1): If $k=0,1,2$, it is clear that
$y_{_{\delta_3,k}}=2^k\delta_3\geq\delta_3$, then

$n-y_{_{\delta_3,k}}-\delta_3=n-(2^k+1)\delta_3\geq
\frac{n+15}{6}+10>0$.

(2.2): If $k=3,5,\cdots,2t-3$, then we can obtain that
\begin{eqnarray*}
y_{_{\delta_3,k}}&=&2^k(\delta_2-2)-\frac{2^{k-1}-1}{3}n=2\delta_2-2^{k+1}-2^{k-1}+1,\\
y_{_{\delta_3,k}}-\delta_3&=&\delta_2-2^{k+1}-2^{k-1}+3\\
&\geq& \delta_2-2^{(2t-3)+1}-2^{(2t-3)-1}+3=
\delta_2-2^{2t-2}-2^{2t-4}+3>0,\\
n-y_{_{\delta_3,k}}-\delta_3&=&3\delta_2+2^{k+1}+2^{k-1}+4>0.
\end{eqnarray*}

(2.3): If $k=4,6,\cdots,2t-2$,    we can deduce that
\begin{eqnarray*}
y_{_{\delta_3,k}}&=&2^k(\delta_2-2)-\frac{2^{k-1}-2}{3}n=4\delta_2-2^{k+1}-2^{k-1}+2,\\
y_{_{\delta_3,k}}-\delta_3&=&3\delta_2-2^{k+1}-2^{k-1}+4\\
&\geq&\!\!3\delta_2-2^{(2t-2)+1}-2^{(2t-2)-1}+4=3\delta_2-2^{2t-1}-2^{2t-3}+4>0,\\
n-y_{_{\delta_3,k}}-\delta_3&=&\delta_2+2^{k+1}+2^{k-1}+3>0.
\end{eqnarray*}

(2.4): If $k=2t-1$, it is easy to know
\begin{eqnarray*}
y_{_{\delta_3,k}}&=&2^k(\delta_2-2)-(\frac{2^{k-1}-1}{3}-1)n=n+2\delta_2-2^{k+1}-2^{k-1}+1,\\
y_{_{\delta_3,k}}-\delta_3&=&n+\delta_2-2^{k+1}-2^{k-1}+3=\delta_2+2^{2t}-2^{2t-2}+4>0,\\
n-y_{_{\delta_3,k}}-\delta_3&=&2^{k+1}+2^{k-1}+1-3\delta_2=2^{2t-2}+2>0.
\end{eqnarray*}

(2.5): If $k=2t$, then we can check that
\begin{eqnarray*}
y_{_{\delta_3,k}}&=&2^k(\delta_2-2)-(\frac{2^{k-1}-2}{3}-1)n=4\delta_2-2^{k-1}+3,\\
y_{_{\delta_3,k}}-\delta_3&=&3\delta_2-2^{k-1}+5=2^{2t-1}+4>0,\\
n-y_{_{\delta_3,k}}-\delta_3&=&\delta_2+2^{k-1}+2>0.
\end{eqnarray*}

Summarizing the previous five cases, we know
$\delta_3$ is a coset leader.\\

{\it Step 3:} We show $y_{_{\delta_4,k}}-\delta_4\geq 0$ and
$n-y_{_{\delta_4,k}}-\delta_4\geq 0$ in seven cases:

(3.1): If $k=0,1,2$, it is clear that
$y_{_{\delta_4,k}}=2^k\delta_4\geq \delta_4$, then

$n-y_{_{\delta_4,k}}=n-2^k(\delta_2-8)=(6-2^k)\delta_2+2^{k+3}+3
>\delta_4$.

(3.2): If $k=3,5,\cdots,2t-5$,  we have then
\begin{eqnarray*}
y_{_{\delta_4,k}}&=&2^k(\delta_2-8)-\frac{2^{k-1}-1}{3}n=2\delta_2-2^{k+3}-2^{k-1}+1,\\
y_{_{\delta_4,k}}-\delta_4&=&\delta_2-2^{k+3}-2^{k-1}+9\\
&\geq&\delta_2-2^{(2t-5)+3}-2^{(2t-5)-1}+9=\delta_2-2^{2t-2}-2^{2t-6}+9>0,\\
n-y_{_{\delta_4,k}}-\delta_4&=&3\delta_2+2^{k+3}+2^{k-1}+10>0.
\end{eqnarray*}

(3.3): If $k=4,6,\cdots,2t-4$,  we can check that
\begin{eqnarray*}
y_{_{\delta_4,k}}&=&2^k(\delta_2-8)-\frac{2^{k-1}-2}{3}n=4\delta_2-2^{k+3}-2^{k-1}+2,\\
y_{_{\delta_4,k}}-\delta_4&=&3\delta_2-2^{k+3}-2^{k-1}+10\\
&\geq&3\delta_2-2^{(2t-4)+3}-2^{(2t-4)-1}+10=2^{2t-1}-2^{2t-5}+9>0,\\
n-y_{_{\delta_4,k}}-\delta_4&=&\delta_2+2^{k+3}+2^{k-1}+9>0.
\end{eqnarray*}

(3.4): If $k=2t-3$, then we deduce that
\begin{eqnarray*}
y_{_{\delta_4,k}}&=&2^k(\delta_2-8)-(\frac{2^{k-1}-1}{3}-1)n=5\delta_2-2^{k-1}+3,\\
y_{_{\delta_4,k}}-\delta_4&=&4\delta_2-2^{k-1}+11=\delta_2+2^{2t}-2^{2t-4}+10>0,\\
n-y_{_{\delta_4,k}}-\delta_4&=&2^{k-1}+8>0.
\end{eqnarray*}

(3.5): If $k=2t-2$,  we obtain that
\begin{eqnarray*}
y_{_{\delta_4,k}}&=&2^k(\delta_2-8)-(\frac{2^{k-1}-2}{3}-1)n=4\delta_2-2^{k-1}+3,\\
y_{_{\delta_4,k}}-\delta_4&=&3\delta_2-2^{k-1}+11=2^{2t}-2^{2t-3}+10>0,\\
n-y_{_{\delta_4,k}}-\delta_4&=&\delta_2+2^{k-1}+8>0.
\end{eqnarray*}

(3.6): If $k=2t-1$, then we get
\begin{eqnarray*}
y_{_{\delta_4,k}}&=&2^k(\delta_2-8)-(\frac{2^{k-1}-1}{3}-2)n=2\delta_2-2^{k-1}+3,\\
y_{_{\delta_4,k}}-\delta_4&=&\delta_2-2^{k-1}+11=\delta_2-2^{2t-2}+11>0,\\
n-y_{_{\delta_4,k}}-\delta_4&=&3\delta_2+2^{k-1}+8>0.
\end{eqnarray*}

(3.7): If $k=2t$, it is easy to derive that
\begin{eqnarray*}
y_{_{\delta_4,k}}&=&2^k(\delta_2-8)-(\frac{2^{k-1}-2}{3}-4)n=4\delta_2-2^{k-1}+6,\\
y_{_{\delta_4,k}}-\delta_4&=&3\delta_2-2^{k-1}+14=2^{2t-1}+13>0,\\
n-y_{_{\delta_4,k}}-\delta_4&=&\delta_2+2^{k-1}+5>0.
\end{eqnarray*}

It follows from the seven cases above that $\delta_4$ is a coset leader.\\

{\it Step 4:} We show $y_{_{\delta_5,k}}-\delta_5\geq 0$ and
$n-y_{_{\delta_5,k}}-\delta_5\geq 0$ in seven cases:

(4.1): If $k=0,1,2$, it is clear that
$y_{_{\delta_5,k}}=2^k\delta_5\geq \delta_5$, then

$n-y_{_{\delta_5,k}}=n-2^k(\delta_2-10)=(6-2^k)\delta_2+2^{k+3}+2^{k}+3>\delta_5$.

(4.2): If $k=3,5,\cdots,2t-5$, then we can check that
\begin{eqnarray*}
y_{_{\delta_5,k}}&=&2^k(\delta_2-10)-\frac{2^{k-1}-1}{3}n=2\delta_2-2^{k+3}-2^{k+1}-2^{k-1}+1,\\
y_{_{\delta_5,k}}-\delta_5&=&\delta_2-2^{k+3}-2^{k+1}-2^{k-1}+11\\
                          &\geq&\delta_2-2^{(2t-5)+3}-2^{(2t-5)+1}-2^{(2t-5)-1}+11\\
                          &=&\delta_2-2^{2t-2}-2^{2t-4}-2^{2t-6}+11>0,\\
n-y_{_{\delta_5,k}}-\delta_5&=&3\delta_2+2^{k+3}+2^{k+1}+2^{k-1}+12
>0.
\end{eqnarray*}

(4.3): If $k=4,6,\cdots,2t-4$, then we deduce
\begin{eqnarray*}
y_{_{\delta_5,k}}&=&2^k(\delta_2-10)-\frac{2^{k-1}-2}{3}n=4\delta_2-2^{k+3}-2^{k+1}-2^{k-1}+2,\\
y_{_{\delta_5,k}}-\delta_5&=&3\delta_2-2^{k+3}-2^{k+1}-2^{k-1}+12\\
                          &\geq&3\delta_2-2^{(2t-4)+3}-2^{(2t-4)+1}-2^{(2t-4)-1}+12\\
                          &=&2^{2t-1}-2^{2t-3}-2^{2t-5}+11>0,\\
n-y_{_{\delta_5,k}}-\delta_5&=&\delta_2+2^{k+3}+2^{k+1}+2^{k-1}+11>0.
\end{eqnarray*}

(4.4): If $k=2t-3$, it is not difficult to get that
\begin{eqnarray*}
y_{_{\delta_5,k}}&=&2^k(\delta_2-10)-(\frac{2^{k-1}-1}{3}-1)n=5\delta_2-2^{k+1}-2^{k-1}+3,\\
y_{_{\delta_5,k}}-\delta_5&=&4\delta_2-2^{k+1}-2^{k-1}+13\\
&=&\delta_2+2^{2t}-2^{2t-2}-2^{2t-4}+12>0,\\
n-y_{_{\delta_5,k}}-\delta_5&=&2^{k+1}+2^{k-1}+10>0.
\end{eqnarray*}

(4.5): If $k=2t-2$, we can easily obtain
\begin{eqnarray*}
y_{_{\delta_5,k}}&=&2^k(\delta_2-10)-(\frac{2^{k-1}-2}{3}-1)n=4\delta_2-2^{k+1}-2^{k-1}+3,\\
y_{_{\delta_5,k}}-\delta_5&=&3\delta_2-2^{k+1}-2^{k-1}+13=2^{2t-1}-2^{2t-3}+12>0,\\
n-y_{_{\delta_5,k}}-\delta_5&=&\delta_2+2^{k+1}+2^{k-1}+10>0.
\end{eqnarray*}

(4.6): If $k=2t-1$, then we get
\begin{eqnarray*}
y_{_{\delta_5,k}}&=&2^k(\delta_2-10)-(\frac{2^{k-1}-1}{3}-3)n=5\delta_2-2^{k-1}+5,\\
y_{_{\delta_5,k}}-\delta_5&=&4\delta_2-2^{k-1}+15=\delta_2+2^{2t}-2^{2t-2}+14>0,\\
n-y_{_{\delta_5,k}}-\delta_5&=&2^{k-1}+8>0.
\end{eqnarray*}

(4.7): If $k=2t$, then we deduce
\begin{eqnarray*}
y_{_{\delta_5,k}}&=&2^k(\delta_2-10)-(\frac{2^{k-1}-2}{3}-5)n=4\delta_2-2^{k-1}+7,\\
y_{_{\delta_5,k}}-\delta_5&=&3\delta_2-2^{k-1}+17=2^{2t-1}+16>0,\\
n-y_{_{\delta_5,k}}-\delta_5&=&\delta_2+2^{k-1}+6>0.
\end{eqnarray*}
Then we can conclude that $\delta_5$ is a coset leader from the
previous seven cases.
\end{proof}

\subsection*{Appendix C: The proof of Theorem 3.7}

\begin{proof}(1) Let $T_{\delta}$ be the defining set of
$\mathcal{C}(n,2,\delta,1)$ and $T_{\delta}=\bigcup\limits_{i\in
S_{\delta}}C_{i}$, where
 $S_{\delta}=\{x|x~\hbox{is a coset leader,
} C_x\subseteq T_{\delta}\}$, then
  $\mathcal{C}(n,2,\delta,1)$ has
dimension $k=n-|T_{\delta}|=n-\sum\limits_{i\in S_{\delta}}|C_{i}|$.

(i): When $2^{t+1}+3\leq \delta \leq 2^{t+1}+2^{t}-3$, from Theorem
3.1, we have  $$S_{\delta}=\{x|x~\hbox{is odd and}~x\in
[1,\delta-1]\setminus\{2^{t+1}\pm1\}\},$$ thus
$|S_{\delta}|=\frac{\delta-1}{2}-2$. According to Lemma 3.6, all
cyclotomic cosets in $T_{\delta}$ have cardinality $2m$,  it then
follows that $$k=n-\sum\limits_{i\in S_{\delta}}|C_{i}|=n-2m\cdot
(\frac{\delta-1}{2}-2)=n-m\delta+5m.$$  It is obvious that  there
exist $\delta-1$ consecutive integers,  according to the BCH bound,
the minimum distance $d\geq \delta$ .

(ii): Similar to (i),
 when $2^{t+1}+2^{t}+3\leq \delta\leq 2^{t+2}-9$,  from Theorem 3.1 and Lemma 3.6,
  $$S_{\delta}=\{x|x~\hbox{is odd and}~x\in
[1,\delta-1]\setminus\{2^{t+1}\pm1,2^{t+1}+2^{t}\pm1\}\},$$  then we
have  $d\geq \delta$ and $k=n-2m\cdot
(\frac{\delta-1}{2}-4)=n-m\delta+9m.$

(iii): Similar to (i),
 when $2^{t+2}-7\leq\delta\leq 2^{t+2}+9$,  from Theorem 3.1 and Lemma 3.6,
 $$S_{\delta}=\{x|x~\hbox{is odd and}~x\in
[1,2^{t+2}-9]\setminus\{2^{t+1}\pm1,2^{t+1}+2^{t}\pm1\}\},$$  then
we have $d\geq 2^{t+2}+9$ and $k=n-2m\cdot
(2^{t+1}-8)=n-2^{t+2}m+16m.$

(iv): When $\delta_{i+1}+2\leq\delta\leq \delta_{i}(i=1,2,3,4)$, we
can infer from Theorem 3.5 that
 thus $T_{\delta}=\bigcup\limits_{i\in
S_{\delta}}C_{i}=\{1,2,\cdots,n-1\}\setminus\bigcup
\limits_{j=1}^{i}C_{\delta_j}$.

Since $|C_{\delta_{1}}|=2$ from the proof of Lemma 3.2, combining
Lemma 3.6, every $C_{\delta_{i}}(i=2,3,4,5)$ has cardinality $2m$,
it then follows that
$$k=n-|T_{\delta}|=n-[n-1-2m (i-1)-2]=2m (i-1)+3.$$

On the other hand, there exist $\delta_{i}-1$ consecutive integers
in $T_{\delta}$, the minimum distance $d\geq \delta_{i}$.

(v): When $\delta_{1}+2\leq\delta\leq n$, it is easy to infer from
Theorem 3.5 that $T_{\delta}=\bigcup\limits_{i\in
S_{\delta}}C_{i}=\{1,2,\cdots,n-1\}$, then
$k=n-|T_{\delta}|=n-(n-1)=1.$ Obviously, the minimum distance
 $d=n$ by the Singleton bound.

 (2) On the basis of  the proof of (1),   (2) can be
 easily given.
 \end{proof}

\subsection*{Appendix D: The proof of Theorem 4.1}
\begin{proof} Similar to the proof of Theorem 3.1,  it suffices
 to prove  items (2)-(5).

(2) Since $x$ is odd,  $x$ can be denoted by
$$x=2^{2t+1}+1+2l, \hbox{where}~l \in I=[1, 2^{2t}-3].$$

 To verify (2),  we will first determine $y_{_{x,k}}$, then
show $y_{_{x,k}}-x \geq 0$  and $n-y_{_{x,k}}-x \geq 0$ according to
different $k$.

(2.1): If $k=0,1,2,\cdots,2t$, note that $2^{2t+1}+3\leq x\leq
2^{2t+2}-5$,  then $x\leq 2^kx <n$,  it then follows that
$y_{_{x,k}}=2^kx\geq x$, and
\begin{eqnarray*}
n-y_{_{x,k}}-x&=&n-(2^k+1)x\\
&\geq&n-(2^{2t}+1)x\\
 &\geq&n-(2^{2t}+1)(2^{2t+2}-5)=2^{2t}+6>0.
 \end{eqnarray*}
(2.2):  If $k=2t+1$, we check that $y_{_{x,k}}=2^kx-n$, it is easy
to obtain
\begin{eqnarray*}
y_{_{x,k}}-x&=&(2^k-1)x-n\\
&\geq &(2^k-1)(2^{2t+1}+3)-n=2^{2t+2}-4>0,\\
n-y_{_{x,k}}-x&=& 2n-(2^k+1)x\\
&\geq& 2n-(2^k+1)(2^{2t+2}-5)= 3\cdot2^{2t+1}+7>0.
 \end{eqnarray*}
(2.3): For each $k=2t+2, 2t+3,\cdots,4t-1$, it is a little difficult
to determine $y_{_{x,k}}$ and then check   $y_{_{x,k}}-x \geq 0$
along with $n-y_{_{x,k}}-x \geq 0$. To achieve this, we divide
$I=[1, 2^{2t}-3]$ into $2^{k-2t-1}$ subintervals as follows:

 $I_{_{\lambda,k}}=[1,\lambda2^{4t+1-k}-1]$ for $\lambda=1$,

 $I_{_{\lambda,k}}=[(\lambda-1)2^{4t+1-k},
 \lambda2^{4t+1-k}-1]$ for $2\leq \lambda \leq 2^{k-2t-1}-1$,

 $I_{_{\lambda,k}}=[(\lambda-1)2^{4t+1-k}, 2^{2t}-3]$ for $\lambda=2^{k-2t-1}$.

Fix $k$, for each $\lambda \in [1, 2^{k-2t-1}]$,  we  can define
$\lambda$ as the {\it identity tag} of the subinterval
$I_{_{\lambda,k}}=[l_{_{\lambda,b}},l_{_{\lambda,e}}]$.

For given $\lambda$,  $x=2^{2t+1}+1+2l$ with $l\in
I_{_{\lambda,k}}$, we can easily calculate that
$y_{_{x,k}}=2^kx-(2^{k-2t-1}+\lambda-1)n$, then we split into
following two subcases.

{\bf Case 2.3.1:} Firstly, we show $y_{_{x,k}}-x>0$. For general
$\lambda$, we have
\begin{eqnarray*}
y_{_{x,k}}-x&=&(2^k-1)x-(2^{k-2t-1}+\lambda-1)n\\
&=&(2^k-1)(2^{2t+1}+1+2l)-(2^{k-2t-1}+\lambda-1)(2^{4t+2}+1)\\
&=&(2^k-1)(1+2l)-2^{2t+1}-2^{k-2t-1}-(\lambda-1)(2^{4t+2}+1).
\end{eqnarray*}

If $\lambda=1$, then  $l\in I_{_{\lambda,k}}=[1,2^{4t+1-k}-1]$, it
follows from $k\geq 2t+2$ that
\begin{eqnarray*}
y_{_{x,k}}-x&=&(2^k-1)(1+2l)-2^{2t+1}-2^{k-2t-1}\\
&\geq&(2^k-1)(1+2\cdot1)-2^{2t+1}-2^{k-2t-1}\\
&=&(3-2^{-2t-1})\cdot2^k-2^{2t+1}-3\\
&\geq&(3-2^{-2t-1})\cdot2^{2t+2}-2^{2t+1}-3\\
&=&3\cdot2^{2t+2}-2^{2t+1}-5>0.
\end{eqnarray*}
If $\lambda\in[2,2^{k-2t-1}]$, since $l\in
I_{_{\lambda,k}}=[l_{_{\lambda,b}}=(\lambda-1)2^{4t+1-k},
 l_{_{\lambda,e}}]$, we can similarly obtain that
\begin{eqnarray*}
y_{_{x,k}}-x&=&(2^k-1)(1+2l)-2^{2t+1}-2^{k-2t-1}-(\lambda-1)(2^{4t+2}+1)\\
&\geq&(2^k-1)(1+2\cdot l_{_{\lambda,b}})-2^{2t+1}-2^{k-2t-1}-(\lambda-1)(2^{4t+2}+1)\\
&=&2^{4t+2-k}+2^{k}-2^{2t+1}-2^{k-2t-1}-(2^{4t+2-k}+1)\lambda\\
&\geq&2^{4t+2-k}+2^{k}-2^{2t+1}-2^{k-2t-1}-(2^{4t+2-k}+1)\cdot2^{k-2t-1}\\
&=&2^{4t+2-k}+(1-2^{-2t})2^{k}-2^{2t+2}\\
&\geq&2^{4t+2-(2t+2)}+(1-2^{-2t})2^{2t+2}-2^{2t+2}(\hbox{see Lemma 0.1})\\
&=&2^{2t}-4>0.
\end{eqnarray*}
{\bf Case 2.3.2:} Secondly, we show $n-y_{_{x,k}}-x>0$. For general
$\lambda$, we have
\begin{eqnarray*}
n-y_{_{x,k}}-x&=& (2^{k-2t-1}+\lambda)(2^{4t+2}+1)-(2^k+1)(2^{2t+1}+1+2l)\\
& =&\lambda(2^{4t+2}+1)+2^{k-2t-1}-2^{2t+1}-(2^k+1)(1+2l).
\end{eqnarray*}
If $\lambda\in[1,2^{k-2t-1}-1]$  and  $l\in [l_{_{\lambda,b}},
l_{_{\lambda,e}}=\lambda2^{4t+1-k}-1]$, we can derive from $k\geq
2t+2$ that
\begin{eqnarray*}
n-y_{_{x,k}}-x
&\geq&\lambda(2^{4t+2}+1)+2^{k-2t-1}-2^{2t+1}-(2^k+1)(1+2l_{_{\lambda,e}})\\
&=&2^{k-2t-1}+2^k-2^{2t+1}+1-2^{4t+2-k}\lambda\\
&\geq&2^{k-2t-1}+2^k-2^{2t+1}+1-2^{4t+2-k}(2^{k-2t-1}-1)\\
&=&2^{4t+2-k}+(1+2^{-2t})2^k-2^{2t+2}\\
&\geq&2^{4t+2-(2t+2)}+(1+2^{-2t})2^{2t+2}-2^{2t+2}(\hbox{see Lemma 0.1})\\
&=&2^{2t}+4>0.
\end{eqnarray*}

If $\lambda=2^{k-2t-1}$ and $l\in
I_{_{2^{k-2t-1},k}}=[l_{_{\lambda,b}}=(2^{k-2t-1}-1)2^{4t+1-k},l_{_{\lambda,e}}=
2^{2t}-3]$, then we similarly get that
\begin{eqnarray*}
n-y_{_{x,k}}-x&=&2^{k-2t-1}(2^{4t+2}+1)+2^{k-2t-1}-2^{2t+1}-(2^k+1)(1+2l)\\
&\geq&2^{k-2t-1}(2^{4t+2}+1)+2^{k-2t-1}-2^{2t+1}-(2^k+1)(1+2l_{_{\lambda,e}})\\
&=&(5+2^{-2t})2^k-2^{2t+2}+5\\
&\geq&(5+2^{-2t})2^{2t+2}-2^{2t+2}+5=2^{2t+4}+9>0.
\end{eqnarray*}

(2.4): If $k=4t$, similarly, we partition $I=[1, 2^{2t}-3]$ into
$2^{2t-1}-1$ subintervals as follows.

 $I_{_{\lambda,k}}=[1,1=2\lambda-1]$, where $\lambda=1$,

 $I_{_{\lambda,k}}=[2(\lambda-1),
 2\lambda-1]$, where $2\leq \lambda \leq 2^{2t-1}-1$.

Let $\lambda \in [1,2^{2t-1}-1]$, we can define $\lambda$ as the
{\it identity tag} of the subinterval
$I_{_{\lambda,k}}=[l_{_{\lambda,b}},l_{_{\lambda,e}}]$.

For given $\lambda$, if $x=2^{2t+1}+1+2l$ with $l\in
I_{_{\lambda,k}}$, we can easily check that
$y_{_{x,k}}=2^kx-(2^{2t-1}+\lambda-1)n$, then we split into
following two subcases.

{\bf Case 2.4.1:} First, we show $y_{_{x,k}}-x>0$. For general
$\lambda$,
\begin{eqnarray*}
y_{_{x,k}}-x&=&(2^{4t}-1)x-(2^{2t-1}+\lambda-1)n\\
&=&(2^{4t}-1)(2^{2t+1}+1+2l)-(2^{2t-1}+\lambda-1)(2^{4t+2}+1)\\
&=&(2^{4t}-1)(1+2l)-2^{2t+1}-2^{2t-1}-(\lambda-1)(2^{4t+2}+1).
\end{eqnarray*}

If $\lambda=1$, we have $l\in I_{_{1,k}}=[1,1]$, that is $l=1$, we
can easily obtain that $$
y_{_{x,k}}-x=2^{4t}-2^{2t-1}+2^{2t+1}-3>0.$$

If $\lambda\in[2,2^{2t-1}-1]$ and $l\in
[l_{_{\lambda,b}}=2(\lambda-1),
 2\lambda-1]$, we  get
\begin{eqnarray*}
y_{_{x,k}}-x&=&(2^{4t}-1)(1+2l)-2^{2t+1}-2^{2t-1}-(\lambda-1)(2^{4t+2}+1)\\
&\geq&(2^{4t}-1)(1+2l_{_{\lambda,b}})-2^{2t+1}-2^{2t-1}-(\lambda-1)(2^{4t+2}+1)\\
&=&2^{4t}-2^{t+1}-2^{t-1}+4-5\lambda\\
&\geq&2^{4t}-2^{t+1}-2^{t-1}+4-5(2^{2t-1}-1)\\
&=&2^{4t}-5\cdot2^{2t}+9>0.
\end{eqnarray*}
{\bf Case 2.4.2:} Secondly, we show $n-y_{_{x,k}}-x>0$.  For general
$\lambda$,
\begin{eqnarray*}
n-y_{_{x,k}}-x&=&(2^{2t-1}+\lambda)(2^{4t+2}+1)-(2^{4t}+1)(2^{2t+1}+1+2l)\\
& =&\lambda(2^{4t+2}+1)+2^{2t-1}-2^{2t+1}-(2^{4t}+1)(1+2l)\\
&\geq&\lambda(2^{4t+2}+1)+2^{2t-1}-2^{2t+1}-(2^{4t}+1)[1+2(2\lambda-1)]\\
&=&2^{4t}+2^{2t-1}-2^{2t+1}+1-3\lambda\\
&\geq&2^{4t}+2^{2t-1}-2^{2t+1}+1-3(2^{2t-1}-1)\\
&=&2^{4t}-3\cdot 2^{2t}+4>0.
\end{eqnarray*}

(2.5): If $k=4t+1$, for given $l\in I=[1, 2^{2t}-3]$,  we have
$y_{_{x,k}}=2^kx-(2^{2t}+l)n$, thus
\begin{eqnarray*}
y_{_{x,k}}-x&=&(2^k-1)x-(2^{2t}+l)n\\
 &=&2^{4t+1}-2^{2t+1}-2^{2t}-1-3l\\
&\geq &2^{4t+1}-2^{2t+1}-2^{2t}-1-3(2^{2t}-3)\\
&=&2^{4t+1}-3\cdot2^{2t+1}+8>0,\\
n-y_{_{x,k}}-x
&=& 2^{4t+1}+2^{2t}-2^{2t+1}-l\\
&\geq&2^{4t+1}+2^{2t}-2^{2t+1}-(2^{2t}-3)\\
&=& 2^{4t+1}-2^{2t+1}+3>0.
\end{eqnarray*}
Concluding the previous five cases (2.1)-(2.5), (2) holds.

 (3) Since $x$ is odd and $2^{2t+2}+5\leq x\leq 2^{2t+2}+2^{2t}-3$,
   then $x$ can be denoted as
   $$x=2^{2t+2}+1+2l,  \hbox{where}~l\in J=[2,2^{2t-1}-2].$$

 To verify that (3) holds, it is necessary to  show  $y_{_{x,k}}-x \geq 0$  and
$n-y_{_{x,k}}-x \geq 0$. Similar to (2),  we split into following
cases according to different $k$.

(3.1): If $k=0,1,2,\cdots,2t-1$, obviously, $x\leq 2^kx <n$, then
$y_{_{x,k}}=2^kx\geq x$, we get that
\begin{eqnarray*}
n-y_{_{x,k}}-x&=&n-(2^k+1)x\\
&\geq&n-(2^{2t-1}+1)x\\
&\geq&n-(2^{2t-1}+1)(2^{2t+2}+2^{2t}-3)\\
&=&3\cdot2^{4t-1}-7\cdot2^{2t-1}+4>0.
\end{eqnarray*}

(3.2): If $k=2t,2t+1,2t+2$, we check that
$y_{_{x,k}}=2^kx-2^{k-2t}n$, thus
\begin{eqnarray*}
y_{_{x,k}}-x&=&(2^k-1)x-2^{k-2t}n\\
 &=&2^k-2^{k-2t}-2^{2t+2}-1+(2^{k+1}-2)l\\
  &\geq&2^k-2^{k-2t}-2^{2t+2}-1+(2^{k+1}-2)\cdot 2\\
&=&(5-2^{-2t})\cdot2^k-2^{2t+2}-5\\
&\geq&(5-2^{-2t})\cdot2^{2t}-2^{2t+2}-5\\
 &=&2^{2t}-6>0,\\
n-y_{_{x,k}}-x&=&(2^{k-2t}+1)n-(2^k+1)(2^{2t+2}+1+2l)\\
&\geq&(2^{k-2t}+1)n-(2^k+1)[2^{2t+2}+1+2(2^{2t-1}-2)]\\
&=&2^{4t+2}+2^{2t+2}+2^{2t}-2-(2^{2t}+2^{-2t}-3)2^k\\
&\geq&2^{4t+2}+2^{2t+2}+2^{2t}-2-(2^{2t}+2^{-2t}-3)2^{2t+2}\\
&=&7\cdot2^{2t}+8>0.
\end{eqnarray*}

(3.3): For each $k=2t+3, 2t+4,\cdots,4t-1$, it is not easy to
determine $y_{_{x,k}}$ and    check   $y_{_{x,k}}-x \geq 0$ along
with $n-y_{_{x,k}}-x \geq 0$. To achieve this,  we divide the value
range  $J=[2,2^{2t-1}-2]$ of $l$ into $2^{k-2t-2}$ subintervals as
follows:

 $J_{_{\lambda,k}}=[2,\lambda2^{4t+1-k}-1]$ for $\lambda=1$,

 $J_{_{\lambda,k}}=[(\lambda-1)2^{4t+1-k},
 \lambda2^{4t+1-k}-1]$ for $2\leq \lambda \leq 2^{k-2t-2}-1$

$J_{_{\lambda,k}}=[(\lambda-1)2^{4t+1-k}, 2^{2t-1}-2]$ for
$\lambda=2^{k-2t-2}$.

Let $\lambda \in [1, 2^{k-2t-2}]$,  where $\lambda$ is called the
{\it identity tag} of  subinterval
$I_{_{\lambda,k}}=[l_{_{\lambda,b}},l_{_{\lambda,e}}]$.

For given $\lambda$, if $x=2^{2t+1}+1+2l$ with $l\in
J_{_{\lambda,k}}$, we can check
$y_{_{x,k}}=2^kx-(2^{k-2t}+\lambda-1)n$,  then we split into
following two subcases:

 {\bf Case 3.3.1:} Firstly, we show
$y_{_{x,k}}-x>0$:
\begin{eqnarray*}
y_{_{x,k}}-x&=&(2^k-1)x-(2^{k-2t}+\lambda-1)n\\
&=&(2^k-1)(2^{2t+2}+1+2l)-(2^{k-2t}+\lambda-1)(2^{4t+2}+1)\\
&=&(2^k-1)(1+2l)-2^{2t+2}-2^{k-2t}-(\lambda-1)(2^{4t+2}+1).
\end{eqnarray*}
If $\lambda=1$, from $l\in J_{_{1,k}}=[2, 2^{4t+1-k}-1]$ and $k\geq
2t+3$, we get
\begin{eqnarray*}
y_{_{x,k}}-x&=&(2^k-1)(1+2l)-2^{2t+2}-2^{k-2t}\\
&\geq&(2^k-1)(1+2\cdot2)-2^{2t+2}-2^{k-2t}\\
&=&(5-2^{-2t})\cdot2^k-2^{2t+2}-5\\
&\geq&(5-2^{-2t})\cdot2^{2t+3}-2^{2t+2}-5\\
&=&9\cdot2^{2t+2}-13>0.
\end{eqnarray*}
If $\lambda\in[2,2^{k-2t-2}]$, thus $l\in J_{_{\lambda,k}}=
[l_{_{\lambda,b}}=(\lambda-1)2^{4t+1-k}, l_{_{\lambda,e}}]$, we can
similarly deduce that
\begin{eqnarray*}
&&y_{_{x,k}}-x\\
&=&(2^k-1)(1+2l)-2^{2t+2}-2^{k-2t}-(\lambda-1)(2^{4t+2}+1)\\
&\geq&(2^k-1)(1+2l_{_{\lambda,b}})-2^{2t+2}-2^{k-2t}-(\lambda-1)(2^{4t+2}+1)\\
&=&2^{4t+2-k}+2^{k}-2^{2t+2}-2^{k-2t}-(2^{4t+2-k}+1)\lambda\\
&\geq&2^{4t+2-k}+2^{k}-2^{2t+2}-2^{k-2t}-(2^{4t+2-k}+1)2^{k-2t-2}\\
&=&2^{4t+2-k}+(1-2^{-2t}-2^{-2t-2})2^{k}-2^{2t+2}-2^{2t}\\
&\geq&2^{4t+2-(2t+3)}+\!\!(1-\!\!2^{-2t}-\!\!2^{-2t-2})2^{2t+3}-2^{2t+2}-2^{2t}(\hbox{see Lemma 0.1})\\
&=&7\cdot2^{2t-1}-10>0.
\end{eqnarray*}
{\bf Case 3.3.2:} Secondly, we show $n-y_{_{x,k}}-x>0$:
\begin{eqnarray*}
n-y_{_{x,k}}-x&=& (2^{k-2t}+\lambda)n-(2^k+1)x\\
&=& (2^{k-2t}+\lambda)(2^{4t+2}+1)-(2^k+1)(2^{2t+2}+1+2l)\\
& =&\lambda(2^{4t+2}+1)+2^{k-2t}-2^{2t+2}-(2^k+1)(1+2l).
\end{eqnarray*}
If $\lambda\in[1,2^{k-2t-2}-1]$, we have $l\in
J_{_{\lambda,k}}=[l_{_{\lambda,b}},
 l_{_{\lambda,e}}=\lambda2^{4t+1-k}-1]$, then
\begin{eqnarray*}
  &&n-y_{_{x,k}}-x\\
  &\geq&\lambda(2^{4t+2}+1)+2^{k-2t}-2^{2t+2}-(2^k+1)(1+2l_{_{\lambda,e}})\\
  &=&2^{k-2t}+2^k-2^{2t+2}+1-2^{4t+2-k}\lambda\\
&\geq&2^{k-2t}+2^k-2^{2t+2}+1-2^{4t+2-k}\cdot(2^{k-2t-2}-1)\\
&=&2^{4t+2-k}+(1+2^{-2t}+2^{-2t-2})2^k-2^{2t+2}-2^{2t}\\
&\geq&2^{4t+2-(2t+3)}+(1+2^{-2t}+2^{-2t-2})2^{2t+3}-2^{2t+2}-2^{2t}(\hbox{see Lemma 0.1})\\
&=&7\cdot2^{2t-1}-10>0;
\end{eqnarray*}
If $\lambda=2^{k-2t-2}$, we have $l\in
J_{_{\lambda,k}}=[l_{_{\lambda,b}}=(\lambda-1)2^{4t+1-k},
 l_{_{\lambda,e}}= 2^{2t-1}-2]$,  it is easy
to obtain that
\begin{eqnarray*}
 n-y_{_{x,k}}-x&
 =&2^{k-2t-2}(2^{4t+2}+1)+2^{k-2t}-2^{2t+2}-(2^k+1)(1+2l)\\
&\geq&2^{k-2t-2}(2^{4t+2}+1)+2^{k-2t}-2^{2t+2}-(2^k+1)(1+2 l_{_{\lambda,e}})\\
&=&(3+2^{-2t}+2^{-2t-2})2^k-2^{2t+2}-2^{2t}+3\\
&\geq&(3+2^{-2t}+2^{-2t-2})2^{2t+3}-2^{2t+2}-2^{2t}+3\\
&=&19\cdot2^{2t}+13>0.
\end{eqnarray*}
(3.4): If $k=4t$, we divide $J=[2,2^{2t-1}-2]$ into the following
$2^{2t-2}-1$ subintervals:

 $J_{_{\lambda,k}}=[2\lambda,  2\lambda+1]$, where $1\leq \lambda \leq 2^{2t-2}-2$,

 $J_{_{\lambda =2^{2t-2}-1,k}}= [2^{2t-1}-2=2\lambda,2^{2t-1}-2=2\lambda]$.

Let $\lambda \in [1,2^{2t-2}-1]$,  we  can define $\lambda$ as the
{\it identity tag} of the subinterval  $J_{_{\lambda,k}}$.

For $x=2^{2t+1}+1+2l$, if $l\in
J_{_{\lambda,k}}=[l_{_{\lambda,b}}=2\lambda,l_{_{\lambda,e}}]$, we
have $y_{_{x,k}}=2^kx-(2^{2t}+\lambda)n$,  then we split into
following two subcases.

{\bf Case 3.4.1:} Firstly, we show $y_{_{x,k}}-x>0$:
\begin{eqnarray*}
y_{_{x,k}}-x&=&(2^{4t}-1)x-(2^{2t}+\lambda)n\\
&=&(2^{4t}-1)(2^{2t+2}+1+2l)-(2^{2t}+\lambda)(2^{4t+2}+1)\\
&=&(2^{4t}-1)(1+2l)-2^{2t+2}-2^{2t}-(2^{4t+2}+1)\lambda\\
&\geq&(2^{4t}-1)(1+2l_{_{\lambda,b}})-2^{2t+2}-2^{2t}-(2^{4t+2}+1)\lambda\\
&=&2^{4t}-2^{2t+2}-2^{2t}-1-5\lambda\\
&\geq&2^{4t}-2^{2t+2}-2^{2t}-1-5(2^{2t-1}-1)\\
&=&2^{4t}-25\cdot2^{2t-2}+4>0.
\end{eqnarray*}

{\bf Case 3.4.2:} Secondly, we show $n-y_{_{x,k}}-x>0$:
\begin{eqnarray*}
n-y_{_{x,k}}-x&=&(2^{2t}+\lambda+1)n-(2^{4t}+1)x\\
&=&(2^{2t}+\lambda+1)(2^{4t+2}+1)-(2^{4t}+1)(2^{2t+2}+1+2l).
\end{eqnarray*}
When $\lambda \in [1,2^{2t-2}-2]$, we have $l\in
J_{_{\lambda,k}}=[l_{_{\lambda,b}}=2\lambda,
l_{_{\lambda,e}}2\lambda+1]$, then
\begin{eqnarray*}
n-y_{_{x,k}}-x&\geq&(2^{2t}+\lambda+1)(2^{4t+2}+1)-(2^{4t}+1)(2^{2t+2}+1+2l_{_{\lambda,e}})\\
&=&2^{4t}-3\cdot2^{2t}-2-3\lambda\\
&\geq&2^{4t}-3\cdot2^{2t}-2-3(2^{2t-2}-2)\\
&=&2^{4t}-15\cdot 2^{2t-2}+4>0.
\end{eqnarray*}
When $\lambda=2^{2t-2}-1$ and $l\in
J_{_{\lambda,k}}=[2^{2t-1}-2,2^{2t-1}-2]$, we easily deduce that
\begin{eqnarray*}
n-y_{_{x,k}}-x&=&(2^{2t}+\lambda+1)(2^{4t+2}+1)-(2^{4t}+1)(2^{2t+2}+1+2l)\\
&\geq&(2^{2t}+2^{2t-2})(2^{4t+2}+1)-(2^{4t}+1)[2^{2t+2}+1+2(2^{2t-1}-2)]\\
&=&3\cdot 2^{4t}-15\cdot 2^{2t-2}+3>0.
\end{eqnarray*}
(3.5): If  $k=4t+1$, for all $x=2^{2t+2}+1+2l$ with $l\in
J=[2,2^{2t-1}-2]$, we have $y_{_{x,k}}=2^kx-(2^{k-2t}+l)n$, it
follows that
\begin{eqnarray*}
y_{_{x,k}}-x&=&(2^k-1)x-(2^{k-2t}+l)n\\
 &=&2^{4t+1}-3\cdot2^{2t+1}-1-3l\\
&\geq &2^{4t+1}-3\cdot2^{2t+1}-1-3(2^{2t-1}-2)\\
 &=&2^{4t+1}-15\cdot2^{2t-1}+5>0,\\
&&\\
n-y_{_{x,k}}-x&=& (2^{2t}+l+1)n-(2^k+1)x\\
&=& 2^{4t+1}-2^{2t+1}-l\\
&\geq &2^{4t+1}-2^{2t+1}-(2^{2t-1}-2)\\
&=&2^{4t+1}-5\cdot2^{2t+1}+2>0.
\end{eqnarray*}

To conclude the five cases (3.1)-(3.5), (3) holds.

(4) Similarly, for odd $x$ with $2^{2t+2}+2^{2t}+3 \leq x\leq
2^{2t+2}+2^{2t+1}-3$, $x$ can be denoted as
$$x=2^{2t+2}+2^{2t}+1+2l,   \hbox{where}~l \in S=[1,2^{2t-1}-2].$$

 To verify that (4) holds, it is necessary to  show  $y_{_{x,k}}-x \geq 0$  and
$n-y_{_{x,k}}-x \geq 0$.  We now split into following cases
according to different $k$.

(4.1): If $k=0,1,2,\cdots, 2t-1$,  it follows from
$2^{2t+2}+2^{2t}+3\leq x\leq 2^{2t+2}+2^{2t+1}-3$ that $x\leq 2^kx
<n$,   then we easily know $y_{_{x,k}}=2^kx\geq x$ and
 \begin{eqnarray*}
n-y_{_{x,k}}-x&=&n-(2^k+1)x\\
 &\geq&n-(2^{2t-1}+1)x\\
&\geq&n-(2^{2t-1}+1)(2^{2t+2}+2^{2t+1}-3)\\
&=&2^{4t}-9\cdot2^{2t-1}+4>0.
 \end{eqnarray*}

(4.2): If $k=2t,2t+1$, we check that $y_{_{x,k}}=2^kx-2^{k-2t}n$,
then
\begin{eqnarray*}
y_{_{x,k}}-x&=&(2^k-1)x-2^{k-2t}n\\
 &=&2^k+2^{k+2t}-2^{2t}-2^{k-2t}-2^{2t+2}-1+(2^{k+1}-2)l\\
&\geq&2^k+2^{k+2t}-2^{2t}-2^{k-2t}-2^{2t+2}-1+(2^{k+1}-2)\cdot1\\
&=&(3+2^{2t}-2^{-2t})\cdot2^k-2^{2t+2}-2^{2t}-3\\
&\geq&(3+2^{2t}-2^{-2t})\cdot2^{2t}-2^{2t+2}-2^{2t}-3\\
&=&2^{4t}-2^{2t+1}-4>0,
\end{eqnarray*}
\begin{eqnarray*}
n-y_{_{x,k}}-x&=&(2^{k-2t}+1)n-(2^k+1)(2^{2t+2}+2^{2t}+1+2l)\\
&\geq&(2^{k-2t}+1)n-(2^k+1)[2^{2t+2}+2^{2t}+1+2(2^{2t-1}-2)]\\
&=&2^{4t+2}-2^{2t+2}-2^{2t+1}-(2^{2t+1}-2^{-2t}-3)2^k+4\\
&\geq&2^{4t+2}-2^{2t+2}-2^{2t+1}-(2^{2t+1}-2^{-2t}-3)2^{2t+1}+4\\
&=& 6>0.
\end{eqnarray*}
(4.3): If $k=2t+2$,  we can deduce $y_{_{x,k}}=2^kx-5n$, it follows
that \begin{eqnarray*}
y_{_{x,k}}-x&=&(2^k-1)x-5n\\
 &\geq&(2^k-1)(2^{2t+2}+2^{2t}+3)-5n\\
 &=& 7\cdot2^{2t}-8>0,\\
&&\\
n-y_{_{x,k}}-x&=&6n-(2^k+1)x\\
&\geq&6n-(2^k+1)(2^{2t+2}+2^{2t+1}-3)\\
 &=& 3(2^{2t+1}+3)>0.
\end{eqnarray*}

(4.4): For each $k=2t+3, 2t+4,\cdots,4t-1$, it is a little
intractable to determine   $y_{_{x,k}}$ and    check $y_{_{x,k}}-x
\geq 0$ along with $n-y_{_{x,k}}-x \geq 0$. To complete this,  we
first divide $S=[1,2^{2t-1}-2]$ into $2^{k-2t-2}$ intervals as
follows:

 $S_{_{\lambda,k}}=[1,2^{4t+1-k}-1]$ for $\lambda=1$,

 $S_{_{\lambda,k}}=[(\lambda-1)2^{4t+1-k},
 \lambda2^{4t+1-k}-1]$ for $2\leq \lambda \leq 2^{k-2t-2}-1$,

 $S_{_{\lambda,k}}=[(\lambda-1)2^{4t+1-k}, 2^{2t-1}-2]$ for $\lambda=2^{k-2t-2}$.

Let $\lambda \in [1,2^{k-2t-2}]$,  we  can define $\lambda$ as the
{\it identity tag} of the subinterval
$S_{_{\lambda,k}}=[l_{_{\lambda,b}},l_{_{\lambda,e}}]$.

For $x=2^{2t+2}+2^{2t}+1+2l$, for given $\lambda$, if $l\in
S_{_{\lambda,k}}$, it is not difficult to check
$y_{_{x,k}}=2^kx-(2^{k-2t}+2^{k-2t-2}+\lambda-1)n$, then we can
split into following two subcases to verify the desired conclusion.

{\bf Case 4.4.1:} Firstly, we show $y_{_{x,k}}-x>0$:
\begin{eqnarray*}
y_{_{x,k}}-x&=&(2^k-1)x-(2^{k-2t}+2^{k-2t-2}+\lambda-1)n\\
&=&(2^k-1)(2^{2t+2}+2^{2t}+1+2l)-(2^{k-2t}+2^{k-2t-2}+\lambda-1)n\\
&=&(2^k-1)(1+2l)-2^{2t+2}-2^{2t}-2^{k-2t}-2^{k-2t-2}-(\lambda-1)n.
\end{eqnarray*}

If $\lambda=1$, we have $l\in S_{_{1,k}}=[1,2^{4t+1-k}-1]$, it
follows from $k\geq 2t+3$ that
 \begin{eqnarray*}y_{_{x,k}}-x
&=&(2^k-1)(1+2l)-2^{2t+2}-2^{2t}-2^{k-2t}-2^{k-2t-2}\\
&\geq&(2^k-1)(1+2cdot1)-2^{2t+2}-2^{2t}-2^{k-2t}-2^{k-2t-2}\\
&=&(3-2^{-2t}-2^{-2t-2})\cdot2^k-2^{2t+2}-2^{2t}-3\\
&\geq&(3-2^{-2t}-2^{-2t-2})\cdot2^{2t+3}-2^{2t+2}-2^{2t}-3\\
&=&19\cdot2^{2t}-13>0.
\end{eqnarray*}

If $\lambda\in[2,2^{k-2t-2}]$, since  $l\in
S_{_{\lambda,k}}=[l_{_{\lambda,b}}=(\lambda-1)2^{4t+1-k},l_{_{\lambda,e}}]$,
we can similarly deduce  from $k\geq 2t+3$ that
\begin{eqnarray*}
y_{_{x,k}}-x&=&(2^k-1)(1+2l)-2^{2t+2}-2^{2t}-2^{k-2t}-2^{k-2t-2}-(\lambda-1)n\\
&\geq&(2^k-1)(1+2l_{_{\lambda,b}})-2^{2t+2}-2^{2t}-2^{k-2t}-2^{k-2t-2}-(\lambda-1)n\\
&=&2^{4t+2-k}+2^{k}-2^{2t+2}-2^{2t}-2^{k-2t}-2^{k-2t-2}-(2^{4t+2-k}+1)\lambda\\
&\geq&\!\!2^{4t+2-k}+2^{k}-2^{2t+2}-2^{2t}-2^{k-2t}-\!2^{k-2t-2}\!\!-(2^{4t+2-k}+1)2^{k-2t-2}\\
&=&2^{4t+2-k}+(1-2^{-2t}-2^{-2t-1})2^{k}-2^{2t+2}-2^{2t+1}(\hbox{see Lemma 0.1})\\
&\geq&2^{4t+2-(2t+3)}+(1-2^{-2t}-2^{-2t-1})2^{2t+3}-2^{2t+2}-2^{2t+1}\\
&=&5\cdot2^{2t-1}-12>0.
\end{eqnarray*}
{\bf Case 4.4.2:} Secondly, we show $n-y_{_{x,k}}-x>0$:
\begin{eqnarray*}
&&n-y_{_{x,k}}-x\\
&=& (2^{k-2t}+2^{k-2t-2}+\lambda)n-(2^k+1)x\\
&=&\lambda(2^{4t+2}+1)+2^{k-2t}+2^{k-2t-2}-\!\!2^{2t+2}-2^{2t}-\!\!(2^k+1)(1+2l).
\end{eqnarray*}
If $\lambda\in[1,2^{k-2t-2}-1]$, we  have $l\in
S_{_{\lambda,k}}=[l_{_{\lambda,b}},
l_{_{\lambda,e}}=\lambda2^{4t+1-k}-1]$, then
\begin{eqnarray*}
 &&n-y_{_{x,k}}-x\\
&\geq&\lambda(2^{4t+2}+1)+2^{k-2t}+2^{k-2t-2}-2^{2t+2}-2^{2t}-(2^k+1)(1+2l_{_{\lambda,e}})\\
 &=&2^{k-2t}+2^{k-2t-2}+2^k-2^{2t+2}-2^{2t}+1-(2^{4t+2-k}-1)\lambda\\
&\geq&2^{k-2t}+2^{k-2t-2}+2^k-2^{2t+2}-2^{2t}+1-(2^{4t+2-k}-1)(2^{k-2t-2}-1)\\
&=&2^{4t+2-k}+(1+2^{-2t}+2^{-2t-1})2^k-2^{2t+2}-2^{2t+1}\\
&\geq&2^{4t+2-(2t+3)}+(1+2^{-2t}+2^{-2t-1})2^{2t+3}-2^{2t+2}-2^{2t+1}(\hbox{see Lemma 0.1})\\
&=&5\cdot2^{2t-1}+12>0.
\end{eqnarray*}

If $\lambda=2^{k-2t-2}$,  it follows that $l\in
S_{_{\lambda,k}}=[l_{_{\lambda,b}}=(2^{k-2t-2}-1)2^{4t+1-k},
l_{_{\lambda,e}}=2^{2t-1}-2]$, then we obtain
\begin{eqnarray*}
&& n-y_{_{x,k}}-x\\
 &=&2^{k-2t-2}(2^{4t+2}+1)+2^{k-2t}+\!\!2^{k-2t-2}-\!\!2^{2t+2}-\!\!2^{2t}-(2^k+1)(1+2l)\\
&\geq&2^{k-2t-2}(2^{4t+2}+1)+2^{k-2t}+\!\!2^{k-2t-2}-\!\!2^{2t+2}-\!\!2^{2t}-\!\!(2^k+1)(1+2l_{_{\lambda,e}})\\
 &=&2^{k-2t-1}+2^{k-2t}+2^{k+1}-2^{2t+2}-2^{2t+1}+3\\
&\geq&2^{(2t+3)-2t-1}+2^{(2t+3)-2t}+2^{(2t+3)+1}-2^{2t+2}-2^{2t+1}+3\\
 &=&5(2^{2t+1}+3)>0.
\end{eqnarray*}
(4.5): If $k=4t$,  similar to (4.4) above, to determine $y_{_{x,k}}$
and    check   $y_{_{x,k}}-x \geq 0$ along with $n-y_{_{x,k}}-x \geq
0$,  we  first divide $S=[1,2^{2t-1}-2]$ into $2^{2t-2}$ intervals
as follows:

 $S_{_{\lambda,k}}=[1,1=2\lambda-1]$ for $\lambda=1$,

 $S_{_{\lambda,k}}=[2(\lambda-1),
 2\lambda-1] ~ \mbox{for} ~ 2\leq \lambda \leq 2^{2t-2}-1 $,

 $S_{_{\lambda,k}}=[2^{2t-1}-2=2(\lambda-1),2^{2t-1}-2]$ for $\lambda=2^{2t-2}$.

Let $\lambda \in [1,2^{2t-2}]$,  we  can define $\lambda$ as the
{\it identity tag} of the subinterval
$S_{_{\lambda,k}}=[l_{_{\lambda,b}},l_{_{\lambda,e}}]$.

For give $\lambda$, if $x=2^{2t+2}+2^{2t}+1+2l$ with $l\in
S_{_{\lambda,k}}$, it is not difficult to check that
$y_{_{x,k}}=2^kx-(2^{k-2t}+2^{k-2t-2}+\lambda-1)n$,  then we split
into following two subcases.

{\bf Case 4.5.1:} Firstly, we show $y_{_{x,k}}-x>0$:
\begin{eqnarray*}
y_{_{x,k}}-x&=&(2^{4t}-1)x-(2^{2t}+2^{2t-2}+\lambda-1)n\\
&=&(2^{4t}-1)(2^{2t+2}+2^{2t}+1+2l)-(2^{2t}+2^{2t-2}+\lambda-1)(2^{4t+2}+1)\\
&=&(2^{4t}-1)(1+2l)-2^{2t+2}-2^{2t+1}-2^{2t-2}-(2^{4t+2}+1)(\lambda-1).
\end{eqnarray*}
If $\lambda=1$,  we have $l\in S_{_{\lambda,k}}=[1,1]$, that is,
$l=1$, we can easily deduce
\begin{eqnarray*}
y_{_{x,k}}-x&=&(2^{4t}-1)(1+2l)-2^{2t+2}-2^{2t+1}-2^{2t-2}\\
&=&(2^{4t}-1)(1+2\cdot1)-2^{2t+2}-2^{2t+1}-2^{2t-2}\\
&\geq& 3\cdot2^{4t}-25\cdot 2^{2t-2}-3>0.
\end{eqnarray*}
If $\lambda\in[2,2^{2t-2}]$,  we know $l\in
S_{_{\lambda,k}}=[l_{_{\lambda,b}}=2\lambda-2,l_{_{\lambda,e}}]$, it
follows that
\begin{eqnarray*}
y_{_{x,k}}-x&=&(2^{4t}-1)(1+2l)-2^{2t+2}-2^{2t+1}-2^{2t-2}-(2^{4t+2}+1)\\
&\geq&(2^{4t}-1)(1+2l_{_{\lambda,b}})-2^{2t+2}-2^{2t+1}-2^{2t-2}-(2^{4t+2}+1)\\
&=&2^{4t}-25\cdot2^{2t-2}+4-5\lambda\\
&\geq&2^{4t}-25\cdot2^{2t-2}+4-5\cdot2^{2t-2}\\
&=&2^{4t}-15\cdot2^{2t-1}+4>0.
\end{eqnarray*}
{\bf Case 4.5.2:} Secondly, we show $n-y_{_{x,k}}-x>0$:
\begin{eqnarray*}
n-y_{_{x,k}}-x&=&(2^{2t}+2^{2t-2}+\lambda)n-(2^{4t}+1)x\\
&=&(2^{2t}+2^{2t-2}+\lambda)(2^{4t+2}+1)-(2^{4t}+1)(2^{2t+2}+2^{2t}+1+2l).
\end{eqnarray*}
If $1\leq \lambda \leq 2^{2t-2}-1$ and  $l\in
S_{_{\lambda,k}}=[l_{_{\lambda,b}},l_{_{\lambda,e}}=2\lambda-1]$,
then we get that
\begin{eqnarray*}
n-y_{_{x,k}}-x&\geq&(2^{2t}+2^{2t-2}+\lambda)(2^{4t+2}+1)-(2^{4t}+1)(2^{2t+2}+2^{2t}+1+2l_{_{\lambda,e}})\\
&=&2^{4t}+2^{2t-2}-2^{2t+2}+1-3\lambda\\
&\geq&2^{4t}+2^{2t-2}-2^{2t+2}+1-3(2^{2t-2}-1)\\
 &=&2^{4t}-9\cdot2^{2t-1}+4>0.
\end{eqnarray*}
If $\lambda=2^{2t-2}$, we have $x=2^{2t+2}+2^{2t}+1+2l$ with $l\in
S_{_{\lambda,k}}$, it is easy to  deduce that
\begin{eqnarray*}
n-y_{_{x,k}}-x&=&(2^{2t}+2^{2t-2}+2^{2t-2})(2^{4t+2}+1)-(2^{4t}+1)x\\
&\geq&(2^{2t}+2^{2t-2}+2^{2t-2})(2^{4t+2}+1)-(2^{4t}+1)(2^{2t+2}+2^{2t+1}-3)\\
&=&3\cdot 2^{4t}-9\cdot 2^{2t-1}+3>0.
\end{eqnarray*}
(4.6): If $k=4t+1$, for all $x=2^{2t+2}+2^{2t}+1+2l$ with $l\in
S=[1,2^{2t-1}-2]$, we have
$y_{_{x,k}}=2^kx-(2^{k-2t}+2^{k-2t-2}+l)n$, then
\begin{eqnarray*}
y_{_{x,k}}-x&=&(2^k-1)x-(2^{k-2t}+l)n\\
 &=&2^{4t+1}-3\cdot2^{2t+1}-1-3l\\
&\geq &2^{4t+1}-3\cdot2^{2t+1}-1-3(2^{2t-1}-2)\\
  &=&2^{4t+1}-15\cdot2^{2t-1}+5>0,\\
  &&\\
n-y_{_{x,k}}-x&=& (2^{2t}+l+1)n-(2^k+1)x\\
&=& 2^{4t+1}-2^{2t+1}-l\\
&\geq& 2^{4t+1}-2^{2t+1}-(2^{2t-1}-2)\\
&=&2^{4t+1}-5\cdot2^{2t+1}+2>0.
\end{eqnarray*}

Summarizing the six cases (4.1)-(4.6), one can derive that (4)
holds.

(5)  For $i=1\hbox{~or~} 3$, it is easy to check the following:

$(2^{2t+1}+1)2^{2t+1}=2^{m}+2^{t+1}\equiv 2^{2t+1}-1  $;

$(2^{2t+2}-i)2^{6t+2}=(2^{2t+2}-i)2^{2t}\cdot
2^{4t+2}\equiv2^{2t}i+1$;

$(2^{2t+2}+i)2^{2t}=2^{m}+2^{2t}i\equiv 2^{2t}i-1 $;

$(2^{2t+2}+2^{2t}-1)2^{6t+4}=(2^{2t+2}+2^{2t}-1)2^{2t}\cdot 2^{4t+2}
\equiv 2^{2t+2}+5$;

$(2^{2t+2}+2^{2t}+1)2^{2t+2}=2^{m+2}+2^{m}+2^{2t+2}\equiv
2^{t+2}-5$;

$(2^{2t+2}+2^{2t+1}-1)2^{6t+3}=(3\cdot2^{2t+1}-1)2^{2t+1}\cdot
2^{4t+2}\equiv 2^{2t+1}+3$;

$(2^{2t+2}+2^{2t+1}+i)2^{2t+1}=2^{m+1}+2^{m}+2^{2t+1}i\equiv
2^{2t+1}i-3 $.

Combining the definition of a cyclotomic coset, these congruence
expressions above imply that there exists some odd integer
$y\in[1,x-1]$ such that $y\in C_{x}$ for $x$ in (5), hence $x$ is
not a coset leader.
\end{proof}

\subsection*{Appendix E: The proof of Lemma 4.2}

\begin{proof}
 From $\delta_{1}=\frac{n}{5}=2^{4t}-2^{4t-2}+\cdots-2^{2}+1$,
$\delta_{2}=2^{4t-1}+\frac{2^{4t}-1}{5}=2^{4t-1}+3(2^{4t-4}+2^{4t-8}+\cdots+2^{4}+1)$,
we can deduce $\delta_{1}$, $\delta_{2}$, $\delta_{3}$, $\delta_{4}$
and $\delta_{5}$ are all odd. It is easy to derive
 $C_{\delta_{1}}=\{ \delta_{1}, 2\delta_{1}, 4\delta_{1}, 3\delta_{1}
 \}$, which implies that $|C_{\delta_{1}}|=4$ and $\delta_{1}$ is a coset leader.
 Hence,  we will   prove that $\delta_{2}$, $\delta_{3}$, $\delta_{4}$ and $\delta_{5}$
are all coset leaders.

{\it Step 1:} We show $y_{_{\delta_2,k}}-\delta_2\geq 0$ and
$n-y_{_{\delta_2,k}}-\delta_2\geq 0$ in three cases:

(1.1): If $k=0,1,2$, it is clear that

$y_{_{\delta_2,k}}=2^k\delta_2>\delta_2$,
$n-y_{_{\delta_2,k}}-\delta_2=n-(2^k+1)\delta_2 \geq 2^{4t-1}+2>0.$

(1.2): If $k=3,4$, we have
\begin {eqnarray*}
y_{_{\delta_2,k}}&=&2^k\delta_2-2^{k-3}n,\\
5(y_{_{\delta_2,k}}-\delta_2)&=&2^{4t+k}-13\cdot2^{k-3}-7\cdot2^{4t-1}+1\\
&\geq&2^{4t+3}-13\cdot2^{3-3}-7\cdot2^{4t-1}+1=9\cdot2^{4t-1}-12>0,\\
5(n-y_{_{\delta_2,k}}-\delta_2)&=&33\cdot2^{4t-1}+13\cdot2^{k-3}-2^{4t+k}+6\\
&\geq&33\cdot2^{4t-1}+13\cdot2^{4-3}-2^{4t+4}+6=2^{4t-1}+32>0.
 \end {eqnarray*}

(1.3): If $5\leq k\leq 4t+1$, the proof can be split  into following
two subcases.

\quad (1.3.1): When  $k\equiv 1 \bmod 4$ (i.e., $k=5, 9,\cdots,
4t+1$), we have
\begin {eqnarray*}
y_{_{\delta_2,k}}&=&2^k\delta_2-(2^{k-3}+\frac{2^{k-2}-3}{5})n,\\
5(y_{_{\delta_2,k}}-\delta_2)&=&17\cdot2^{4t-1}-2^k-7\cdot2^{k-3}+4\\
&\geq&17\cdot2^{4t-1}-2^{4t+1}-7\cdot2^{(4t+1)-3}+4=
19\cdot2^{4t-2}+4>0,\\
5(n-y_{_{\delta_2,k}}-\delta_2)&=&9\cdot2^{4t-1}\!\!+2^k+7\cdot2^{k-3}+3\\
&\geq&9\cdot2^{4t-1}\!\!+2^5+7\cdot2^{5-3}+3= 9\cdot2^{4t-1}+63>0.
 \end {eqnarray*}
\quad (1.3.2): When  $k\equiv \gamma \bmod 4$ ($\gamma=2,3,4$),
i.e., $k=4\times1+\gamma, 4\times2+\gamma,\cdots, 4(t-1)+\gamma$, it
then follows that
\begin {eqnarray*}
y_{_{\delta_2,k}}&=&2^k\delta_2-(2^{k-3}+\frac{2^{k-2}-2^{\gamma-2}}{5})n,\\
5(y_{_{\delta_2,k}}-\delta_2)&=&(2^{\gamma+1}-7)2^{4t-1}+1+2^{\gamma-2}-15\cdot2^{k-3}\\
&\geq&(2^{\gamma+1}-7)2^{4t-1}+1+2^{\gamma-2}-15\cdot2^{[4(t-1)+\gamma]-3}\\
&=&113\cdot 2^{4t+\gamma-7}+2^{\gamma-1}-7\cdot2^{4t-1}+1\\
&\geq&113\cdot 2^{4t+2-7}+2^{2-1}-7\cdot2^{4t-1}+1=2^{4t-5}+2>0,\\
5(n-y_{_{\delta_2,k}}-\delta_2)&=&(33-2^{\gamma+1})2^{4t-1}+6-2^{\gamma-2}+15\cdot2^{k-3}\\
&\geq&(33-2^{\gamma+1})2^{4t-1}+6-2^{\gamma-2}+15\cdot2^{(4\times1+\gamma)-3}\\
&=&33\cdot2^{4t-1}+6-(2^{4t+2}-119)\cdot2^{\gamma-2}\\
&\geq&33\cdot2^{4t-1}+6-(2^{4t+2}-119)\cdot2^{4-2}=2^{4t-1}+482>0.
\end {eqnarray*}

From the three cases above, we then  conclude that $\delta_2$ is a
coset leader.

{\it Step 2:} Now, we show $y_{_{\delta_3,k}}-\delta_3$ and
$n-y_{_{\delta_3,k}}-\delta_3$ for $0\leq k\leq 4t+1$ as follows:

(2.1): If $k=1,2$, it is clear that

$y_{_{\delta_3,k}}=2^k\delta_3>\delta_3$,
$n-y_{_{\delta_3,k}}-\delta_3=n-(2^k+1)\delta_3\geq 2^{4t-1}+32>0$.

(2.2): If $k=3,4$, we obtain
 \begin {eqnarray*}
y_{_{\delta_3,k}}&=&2^k\delta_3-2^{k-3}n,\\
5(y_{_{\delta_3,k}}-\delta_3)&=&2^{4t+k}-253\cdot2^{k-3}-7\cdot2^{4t-1}+31\\
&\geq&2^{4t+3}-253\cdot2^{3-3}-7\cdot2^{4t-1}+31\\
& =&9\cdot 2^{4t-1}-222\geq 930,\\
5(n-y_{_{\delta_3,k}}-\delta_3)&=&33\cdot2^{4t-1}+1+253\cdot2^{k-3}-2^{4t+k}+35\\
&\geq&33\cdot2^{4t-1}+1+253\cdot2^{4-3}-2^{4t+4}+35\\
 &=& 2^{4t-1}  +542>0.
\end {eqnarray*}

(2.3): If $5\leq k\leq 4t-3$, the discussion can be given in two
subcases.

\quad  (2.3.1): When  $k\equiv 1\bmod4$ (i.e., $k=5, 9,\cdots,
4t-3$), we can get
\begin
{eqnarray*}
y_{_{\delta_3,k}}&=&2^k\delta_3-(2^{k-3}+\frac{2^{k-2}-3}{5})n,\\
5(y_{_{\delta_3,k}}-\delta_3)&=&17\cdot2^{4t-1}+4-255\cdot2^{k-3}+30\\
&\geq&17\cdot2^{4t-1}+4-255\cdot2^{(4t-3)-3}+30\\
&=& 289\cdot2^{4t-6}+34>0,\\
5(n-y_{_{\delta_3,k}}-\delta_3)&=&9\cdot2^{4t-1}+255\cdot2^{k-3}+33>0.
\end {eqnarray*}

\quad  (2.3.2): When $t\geq3$ and  $k\equiv \gamma\bmod4$
($\gamma=2,3,4$), i.e., $k=4\times1+\gamma, 4\times2+\gamma,\cdots,
4(t-2)+\gamma$, we obtain
 \begin {eqnarray*}
y_{_{\delta_3,k}}&=&2^k\delta_3-(2^{k-3}+\frac{2^{k-2}-2^{\gamma-2}}{5})n,\\
5(y_{_{\delta_3,k}}\!\!\!-\delta_3)&=&(2^{\gamma+1}-7)\cdot2^{4t-1}+2^{\gamma-2}+1-255\cdot2^{k-3}+30\\
&\geq&(2^{\gamma+1}-7)\cdot2^{4t-1}+2^{\gamma-2}+1-255\cdot2^{[4(t-2)+\gamma]-3}+30\\
&=&1793\cdot2^{4t+\gamma-11}+2^{\gamma}-7\cdot2^{4t-1}+1\\
&\geq&1793\cdot2^{4t+2-11}+2^{2}-7\cdot2^{4t-1}+1\\
&=&2^{4t-9}+32>0,\\
5(n-y_{_{\delta_3,k}}\!\!\!-\delta_3)&=&(33-2^{\gamma+1})2^{4t-1}-2^{\gamma-2}+1+255\cdot2^{k-3}+35\\
&\geq&(33-2^{\gamma+1})2^{4t-1}-2^{\gamma-2}+1+255\cdot2^{(4\times1+\gamma)-3}+35\\
&=&33\cdot2^{4t-1}+1-(2^{4t+2}-2039)\cdot2^{\gamma-2}+35\\
&\geq&33\cdot2^{4t-1}+1-(2^{4t+2}-2039)\cdot2^{4-2}+35\\
&=&2^{4t-1}+8192>0.
\end {eqnarray*}
(2.4): If $4t-2\leq k\leq 4t+1$, the discussion can be given in two
subcases according to different $\gamma$.

\quad  (2.4.1): When  $k=4(t-1)+\gamma$ ($\gamma=2,3,4$), we can
deduce
 \begin {eqnarray*}
y_{_{\delta_3,k}}&=&2^k\delta_3-(2^{k-3}+\frac{2^{k-2}-2^{\gamma-2}}{5}-1)n,\\
5(y_{_{\delta_3,k}}-\delta_3)&=&(2^{\gamma+1}+33)\cdot2^{4t-1}+2^{\gamma-2}+1-255\cdot2^{k-3}+35\\
&=&33\cdot2^{4t}+1-(127\cdot2^{4t+\gamma-7}-2^{\gamma-2})+35\\
&\geq&33\cdot2^{4t}+1-(127\cdot2^{4t+4-7}-2^{4-2})+35\\
&=&5(2^{4t-3}+8)>0,\\
5(n-y_{_{\delta_3,k}}-\delta_3\!\!\!)&=&\!\!\!255\cdot2^{k-3}-(2^{\gamma+1}+7)2^{4t-1}-2^{\gamma-2}+31\\
&=&\!\!\!127\cdot2^{4t+\gamma-7}-2^{\gamma-2}-7\cdot2^{4t-1}+31\\
 &\geq&127\cdot2^{4t+2-7}-2^{2-2}-7\cdot2^{4t-1}+31\\
 &=&15(2^{4t-5}+2)>0.
\end {eqnarray*}

 \quad  (2.4.2): When $k=4t+1$, we have
{  \begin {eqnarray*}
 y_{_{\delta_3,k}}&=&2^k\delta_3-(2^{k-3}+\frac{2^{k-2}-3}{5}-3)n,\\
5(y_{_{\delta_3,k}}-\delta_3)&=&137\cdot2^{4t-1}-255\cdot2^{k-3}+49= 19\cdot2^{4t-2}+49>0,\\
5(n-y_{_{\delta_3,k}}-\delta_3)&=&255\cdot2^{k-3}-111\cdot2^{4t-1}
+18= 33\cdot2^{4t-2}+18>0.
\end {eqnarray*}}

It is easy to know $\delta_1$ is a coset leader,  by  Steps 1  and 2
above, we have shown $\delta_{2}$ and $\delta_{3}$ are coset leaders
respectively. Similar to    Steps 1  and 2, we can also attain
$\delta_4$ and $\delta_5$ are both  coset leaders, the detailed
proofs are omitted.
\end{proof}
\subsection*{Appendix F: The proof of Theorem 4.7}

\begin{proof} Let $\mathcal{C}(n,2,\delta,1)$, $T_{\delta}$ and $S_{\delta}$ be given similar to the proof of Theorem 3.7, then
   $\mathcal{C}(n,2,\delta,1)$ has dimension $k=n-|T_{\delta}|=n-\sum\limits_{i\in S_{\delta}}|C_{i}|$.

(i): When $2^{2t+1}+3\leq \delta\leq 2^{2t+2}-5$,  from Theorem 4.1,
we check that $S_{\delta}=\{x|x~\hbox{is odd and}~x\in
[1,\delta-1]\setminus\{2^{2t+1}+1\}\},$ thus
$|S_{\delta}|=\frac{\delta-1}{2}-1$. By Lemma 4.6, all cyclotomic
cosets in $T_{\delta}$ have cardinality $2m$,  it   follows that $k=
n-2m\cdot (\frac{\delta-1}{2}-1)=n-m\delta+3m $ and $d\geq \delta$.

(ii): Similar to (i),  when $2^{2t+2}+5\leq \delta\leq
2^{2t+2}+2^{2t}-3$, we can derive from Theorem 4.1 and Lemma 4.6
that
$$S_{\delta}=\{x|x~\hbox{is odd and}~x\in
[1,\delta-1]\setminus\{2^{2t+1}+1,2^{2t+2}\pm1,2^{2t+2}\pm3\}\},$$
it follows that $d\geq \delta$  and $k=n-\sum\limits_{i\in
S_{\delta}}|C_{i}|=n-2m\cdot (\frac{\delta-1}{2}-5)=n-m\delta+11m.$

 (iii): Similar to (i), when $2^{2t+2}+2^{2t}+3\leq \delta \leq a-3$,  we can derive from Theorem 4.1 and Lemma 4.6
that $$S_{\delta}=\{x|x~\hbox{is odd and}~x\in
[1,\delta-1]\setminus\{2^{2t+1}+1,2^{2t+2}\pm1,2^{2t+2}\pm3,2^{2t+2}+2^{2t}\pm1\}\},$$
thus  we have $d\geq \delta$ and
   $k=n-\sum\limits_{i\in S_{\delta}}|C_{i}|=n-2m\cdot
(\frac{\delta-1}{2}-7)=n-m\delta+15m.$

(iv):  Similar to (i), when $a-1 \leq \delta\leq a+5$, we can infer
from Theorem 4.5  and Lemma 4.6 that$$S_{\delta}=\{x|x~\hbox{is odd
and}~x\in
[1,a-3]\setminus\{2^{2t+1}+1,2^{2t+2}\pm1,2^{2t+2}\pm3,2^{2t+2}+2^{2t}\pm1\}\},$$
 thus we get $d\geq
a+5$ and   $k=n-2m\cdot (2^{2t+1}+2^{2t}-8)=n-am+16m.$

(v): When $\delta_{i+1}+2\leq\delta\leq \delta_{i}(i=1,2,3,4)$, from
Theorem 4.5, we get that
 $T_{\delta}=\bigcup\limits_{i\in
S_{\delta}}C_{i}=\{1,2,\cdots,n-1\}\setminus\bigcup
\limits_{j=1}^{i}C_{\delta_j}$.

We have known  $|C_{\delta_{1}}|=4$ from the proof of Lemma 4.2,
according to Lemma 4.6, $C_{\delta_{i}}(i=2,3,4,5)$ has cardinality
$2m$, it then follows that
$$k=n-|T_{\delta}|=n-[n-1-2m(i-1)-4]=2m(i-1)+5.$$

On the other hand, there exist $\delta_{i}-1$ consecutive integers
in $T_{\delta}$, the minimum distance $d\geq \delta_{i} $.

(vi): When $\delta_{1}+2\leq\delta\leq n$, we check from Theorem 4.5
that $T_{\delta}=\bigcup\limits_{i\in
S_{\delta}}C_{i}=\{1,2,\cdots,n-1\}$, then
$k=n-|T_{\delta}|=n-(n-1)=1.$ It follows from  Singleton bound that
the minimum distance  $d=n$.

 (2) On the basis of  the proof of (1),   (2) can be
 easily obtained.
 \end{proof}




\end{document}